\documentclass[11pt,draftcls,onecolumn]{IEEEtran}
\usepackage{amsmath,amssymb,eucal,graphicx}
\usepackage{epsfig}
\usepackage{exscale}
\usepackage{latexsym}
\usepackage{verbatim}
\usepackage{amsmath}
\usepackage{amsfonts}
\usepackage{amssymb}
\usepackage{graphicx}%
\usepackage[usenames]{color}
\usepackage[normalem]{ulem}

\newfont{\bbb}{msbm10 scaled 500}

\newfont{\bb}{msbm10 scaled 1100}
\newcommand{\CC}{\mbox{\bb C}}

\newcommand{\gv}{{\bf g}}
\newcommand{\hv}{{\bf h}}

\newcommand{\nv}{{\bf n}}

\newcommand{\uv}{{\bf u}}

\newcommand{\vv}{{\bf v}}
\newcommand{\xv}{{\bf x}}
\newcommand{\yv}{{\bf y}}
\newcommand{\zv}{{\bf z}}
\newcommand{\zerov}{{\bf 0}}

\newcommand{\Am}{{\bf A}}
\newcommand{\Bm}{{\bf B}}

\newcommand{\Gm}{{\bf G}}
\newcommand{\Hm}{{\bf H}}
\newcommand{\Id}{{\bf I}}

\newcommand{\Km}{{\bf K}}

\newcommand{\Sm}{{\bf S}}

\newcommand{\Um}{{\bf U}}

\newcommand{\Vm}{{\bf V}}

\newcommand{\Cc}{{\cal C}}

\newcommand{\Fc}{{\cal F}}

\newcommand{\Kc}{{\cal K}}
\newcommand{\Lc}{{\cal L}}

\newcommand{\Nc}{{\cal N}}

\newcommand{\Rc}{{\cal R}}

\newcommand{\Tc}{{\cal T}}

\newcommand{\Wc}{{\cal W}}

\newcommand{\nuv}{\hbox{\boldmath$\nu$}}

\newcommand{\phiv}{\hbox{\boldmath$\phi$}}

\newcommand{\Gammam}{\hbox{\boldmath$\Gamma$}}
\newcommand{\Lambdam}{\hbox{\boldmath$\Lambda$}}

\newcommand{\Psim}{\hbox{\boldmath$\Psi$}}

\renewcommand{\det}{{\hbox{det}}}
\newcommand{\trace}{{\hbox{tr}}}
\newcommand{\rank}{{\hbox{rank}}}

\newcommand{\SNR}{{\sf SNR}}

\newcommand{\eqdef}{\stackrel{\Delta}{=}}

\newcommand{\cov}{{\hbox{cov}}}

\newtheorem{theorem}{Theorem}
\newtheorem{definition}{Definition}
\newtheorem{lemma}{Lemma}
\newtheorem{cor}{Corollary}[section]

\newtheorem{remark}{\indent \bf Remark}[section] 
\newtheorem{property}{Property}[section]
\def\BibTeX{{\rm B\kern-.05em{\sc i\kern-.025em b}\kern-.08em
T\kern-.1667em\lower.7ex\hbox{E}\kern-.125emX}}

\begin{document}


\title{Secured Communication over Frequency-Selective Fading Channels: \\
a practical Vandermonde precoding\footnote{The material in this paper was partially presented at IEEE 19th International Symposium on Personal, Indoor and Mobile Radio Communications (PIMRC), Cannes, France,  Sept. 2008. }}

\author{\authorblockN{Mari Kobayashi, M\'erouane Debbah}\\
\authorblockA{SUPELEC \\
Gif-sur-Yvette, France\\
Email:\{mari.kobayashi,merouane.debbah\}@supelec.fr}\\
\and \authorblockN{Shlomo Shamai (Shitz)}\\
\authorblockA{Department of Electrical Engineering, Technion-Israel Institute of Technology \\
Haifa, 32000, Israel\\
Email: sshlomo@ee.technion.ac.il\\
} }

\maketitle

\maketitle

\begin{abstract}
In this paper, we study the frequency-selective broadcast channel with confidential messages (BCC) in which the transmitter sends a confidential message to receiver 1 and a common message to receivers 1 and 2.
In the case of  a block transmission of $N$ symbols followed by a guard interval of $L$ symbols, the frequency-selective channel can be modeled as a $N\times (N+L)$ Toeplitz matrix. 
For this special type of multiple-input multiple-output (MIMO) channels, we propose a practical Vandermonde precoding that consists of projecting the confidential messages in the null space of the channel seen by receiver 2 while superposing the common message. For this scheme, we provide the achievable rate region, i.e. the rate-tuple of the common and confidential messages, and characterize the optimal covariance inputs for some special cases of interest.
It is proved that the proposed scheme achieves the optimal degree of freedom (d.o.f) region. More specifically, it enables to send $l\leq L$ confidential messages and $N-l$ common messages simultaneously over a block of $N+L$ symbols.
Interestingly, the proposed scheme can be applied to secured multiuser scenarios such as the $K+1$-user frequency-selective BCC with $K$ confidential messages and the two-user frequency-selective BCC with two confidential messages. For each scenario, we provide the achievable secrecy degree of freedom (s.d.o.f.) region of the corresponding frequency-selective BCC and prove the optimality of the Vandermonde precoding. One of the appealing features of the proposed scheme is that it does not require any specific secrecy encoding
technique but can be applied on top of any existing powerful encoding schemes.

\end{abstract}


\section{Introduction}

We consider a secured medium such that the transmitter wishes to send a confidential message to its receiver while keeping the eavesdropper, tapping the channel, ignorant of the message. Wyner \cite{wyner1975wc} introduced this model named the wiretap channel to model the degraded broadcast channel where the eavesdropper observes a degraded version of the receiver's signal. In this model, the confidentiality is measured by the equivocation rate, i.e. the mutual information between the confidential message and the eavesdropper's observation.
For the discrete memoryless degraded wiretap channel, Wyner characterized the capacity-equivocation region and showed that a non-zero secrecy rate can be achieved \cite{wyner1975wc}. The most important operating point on the capacity-equivocation region is the secrecy capacity, i.e. the largest reliable communication rate such that the eavesdropper obtains no information about the confidential message (the equivocation rate is as large as the message rate). The secrecy capacity of the Gaussian wiretap channel was given in \cite{leungyancheong1978gwt}. Csisz\'ar  and K\"orner considered a more general wiretap channel in which a common message for both receivers is sent in addition to the confidential message \cite{csiszar1978bcc}. For this model known as the broadcast channel with confidential messages (BCC), the rate-tuple of the common and confidential messages was characterized.

Recently, a significant effort has been made to opportunistically exploit the space/time/user dimensions for secrecy communications (see for example \cite{gopala2006scf,khisti2007mc,liu2008scc,liang2007sco,negi2005scu,parada2005scs,khisti2007stm,liu2007nsc,oggier2007scm,shafiee2007arg,khisti:gmw} and references therein). In \cite{gopala2006scf}, the secrecy capacity of the ergodic slow fading channels was characterized and the optimal power/rate allocation was derived. The secrecy capacity of the parallel fading channels was given \cite{liu2008scc,liang2007sco} where
\cite{liang2007sco} considered the BCC with a common message.
Moreover, the secrecy capacity of the wiretap channel with multiple antennas has been studied in \cite{negi2005scu,parada2005scs,khisti2007stm,liu2007nsc,oggier2007scm,shafiee2007arg,immse09} and references therein. In particular, the secrecy capacity of the multiple-input multiple-output (MIMO) wiretap channel has been fully characterized in \cite{khisti:gmw,khisti2007mc,liu2007nsc,oggier2007scm} and more recently its closed-form expressions under a matrix covariance constraint have been derived in \cite{immse09}. Furthermore, a large number of recent works have considered the secrecy capacity region for more general broadcast channels. In \cite{Liangisita08}, the authors studied the two-user MIMO Gaussian BCC where the capacity region for the case of one common and one confidential message was characterized.
The two-user BCC with two confidential messages, each of which must be kept secret to the unintended receiver, has been studied in \cite{liu2008dmi, LiuPoorIT09,LiuLiuPoorShamaiISIT09,choo2008krb}. In \cite{LiuPoorIT09}, Liu and Poor characterized the secrecy capacity region for the multiple-input single-output (MISO) Gaussian BCC where the optimality of the secret dirty paper coding (S-DPC) scheme was proved. A recent contribution \cite{LiuLiuPoorShamaiISIT09} extended the result to the MIMO Gaussian BCC.
The multi-receiver wiretap channels have been also studied in \cite{khisti2008sbo,choo2009trb,ekrem2008scc,ekrem2009scr,bagherikaram2008srr,bagherikaram2009scr} (and reference therein) where the confidential messages to each receiver must be kept secret to an external eavesdropper. It has been proved that the secrecy capacity region of the MIMO Gaussian multi-receiver wiretap channels is achieved by S-DPC \cite{bagherikaram2009scr, ekrem2009scr}.


However, very few work have exploited the frequency selectivity nature of the channel for secrecy purposes \cite{koyluoglu2008ias} 
where the zeros of the channel provide an opportunity to "hide" information. This paper shows the opportunities provided by the broad-band channel and studies the frequency-selective BCC where the transmitter sends one confidential message to receiver 1 and one common message to both receivers 1 and 2.
The channel state information (CSI) is assumed to be known to both the transmitter and the receivers.
We consider the quasi-static frequency-selective fading channel with $L+1$ paths such that the channel remains fixed during an entire transmission of $n$ blocks for an arbitrary large $n$. It should be remarked that in general the secrecy rate cannot scale with signal-to-noise ratio (SNR) over the channel at hand, unless the channel of receiver 2 has a null frequency band of positive Lebesgue measure (on which the transmitter can ``hide'' the confidential message). In this contribution, we focus on the realistic case where receiver 2 has a full frequency band (without null sub-bands) but operates in a reduced dimension due to practical complexity issues. This is typical of current orthogonal frequency division multiplexing (OFDM) standards (such as IEEE802.11a/WiMax or LTE \cite{standard80211,standard80216,standardLTE}) where a guard interval of $L$ symbols is inserted at the beginning of each block to avoid the inter-block interference and both receivers discard these $L$ symbols. We assume that both users have the same standard receiver, in particular receiver 2 cannot change its hardware structure.
Studying secure communications under this assumption is of interest in general, and can be justified since receiver 2 is actually a legitimate receiver which can receive a confidential message in other communication periods. Of course, if receiver 2 is able to access the guard interval symbols, it can extract the confidential message and the secrecy rate falls down to zero. Although we restrict ourselves to the reduced dimension constraint in this paper, other constraints on the limited capability at the unintended receiver such as energy consumption or hardware complexity might provide a new paradigm to design physical layer secrecy systems.

In the case of a block transmission of $N$ symbols followed by a guard interval of $L$ symbols discarded at both receivers, the frequency-selective channel can be modeled as a $N\times (N+L)$ MIMO Toeplitz matrix. In this contribution, we aim at designing a practical linear precoding scheme that fully exploits the degrees of freedom (d.o.f.) offered by this special type of MIMO channels to transmit both the common message and the confidential message.
To this end, let us start with the following remarks. One one hand, the idea of using OFDM modulation to convert the frequency-selective channel represented by the Toeplitz matrix into a set of parallel fading channel turns out to be useless from a secrecy perspective. Indeed, it is known that the secrecy capacity of the parallel wiretap fading channels does not scale
with SNR \cite{liang2007sco}. On the other hand, recent contributions \cite{khisti:gmw,khisti2007mc,liu2007nsc,oggier2007scm,immse09} showed that the secrecy capacity of the MIMO wiretap channel grows linearly with SNR, i.e. $r \log \SNR$ where $r$ denotes the secrecy degree of freedom (s.d.o.f.) (to be specified). In the high SNR regime, the secrecy capacity of the MISO/MIMO wiretap channel is achieved by sending the confidential message in the null space of the eavesdropper's channel \cite{khisti2007stm,khisti:gmw,liu2007nsc,LiuPoorIT09,LiuLiuPoorShamaiISIT09,immse09}. Therefore, OFDM modulation is highly suboptimal in terms of the s.d.o.f..

Inspired by these remarks, we propose a linear Vandermonde precoder that projects the confidential message in the null space of the channel seen by receiver 2 while superposing the common message.
Thanks to the orthogonality between the precoder of the confidential message and the channel of receiver 2, receiver 2 obtains no information on the confidential message. This precoder is regarded as a single-antenna frequency beamformer that nulls the signal in certain directions seen by receiver 2. The Vandermonde structure comes from the fact that the frequency beamformer is of the type $[1,a_i,a_i^2,...,a_i^{N+L}]^T$ where $a_i$ is one of the roots of the channel seen by receiver 2.
Note that Vandermonde matrices \cite{ryandebbah} have already been considered for cognitive radios \cite{paper:sampaiokobayashi2} and CDMA systems \cite{scaglione2000lvm} to reduce/null interference but not for secrecy applications. One of the appealing aspects of Vandermonde precoding is that it does not require a specific secrecy encoding technique but can be applied on top of  any classical capacity achieving encoding scheme.

For the proposed scheme, we characterize its achievable rate region, the rate-tuple
of the common message, the confidential message, respectively. Unfortunately, the optimal input covariances achieving its boundary is generally difficult to compute due to the non-convexity of the weighted sum rate maximization problem.
Nevertheless, we show that there are some special cases of interest such as the secrecy rate and the maximum sum rate point which enable an explicit characterization of the optimal input covariances. In addition, we provide the achievable d.o.f. region of the frequency-selective BCC, reflecting the behavior of the achievable rate region in the high SNR regime, and prove that the Vandermonde precoding achieves this region. More specifically, it enables to simultaneously transmit $l$ streams of the confidential message and $N-l$ streams of the common message for $l\leq L$ simultaneously over a block of $N+L$ dimensions.
Interestingly, the proposed Vandermonde precoding can be applied to multiuser secure communication scenarios; a) a $K+1$-user frequency-selective BCC with $K$ confidential messages and one common message, b) a two-user frequency-selective BCC with two confidential messages and one common message. For each scenario, we characterize the achievable s.d.o.f. region of the corresponding frequency-selective BCC and show the optimality of the Vandermonde precoding.

The paper is organized as follows. Section \ref{sec:model} presents the frequency-selective fading BCC. Section \ref{sec:vandermonde} introduces
the Vandermonde precoding and characterizes its achievable rate region as well as
the optimal input covariances for some special cases. Section \ref{sec:multiuser} provides the application of the Vandermonde precoding
 to the multi-user secure communications scenarios. Section \ref{sec:numerical} some numerical examples
of the proposed scheme in the various settings, and finally Section \ref{sec:conclusion} concludes the paper.

\textbf{Notation : }
In the following, upper (lower boldface) symbols
will be used for matrices (column vectors) whereas lower symbols
will represent scalar values, $(.)^T$ will denote transpose
operator, $(.)^\star$ conjugation and
$(.)^H=\left((.)^T\right)^\star$ hermitian transpose.
 ${\bf I}_n, \zerov_{n\times m}$
represents the $n\times n$ identity matrix, $n\times m$ zero matrix. $|\Am|, \rank(\Am), \trace(\Am)$ denotes a determinant, rank, trace of a matrix $\Am$, respectively. $\xv^n$ denotes the sequence $(\xv[1],\dots,\xv[n])$. $w,u,v,\xv,\yv,\zv$ denotes the realization of the random variables $W,U,V,X,Y,Z$. Finally, $``\preceq''$ denotes less or equal to in the positive semidefinite ordering between positive semidefinite matrices, i.e. we have $\Am \preceq \Bm$ if $\Bm-\Am$ is positive semidefinite.

\section{System Model}\label{sec:model}

\begin{figure}[t]
    \centering
    \includegraphics[width=0.7\columnwidth]{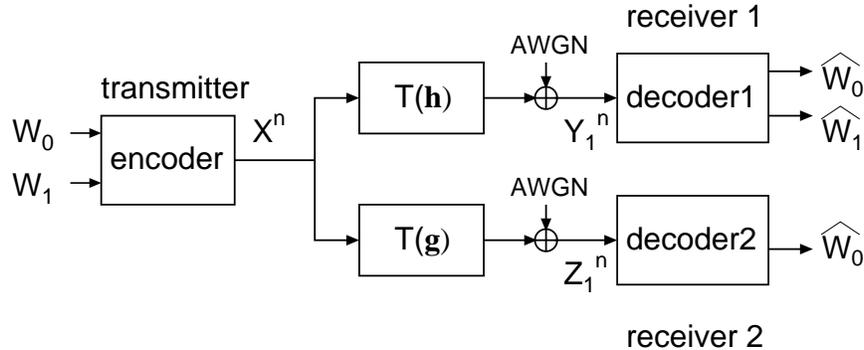}\\
    \vspace{-1em}
    \caption{Frequency-selective broadcast channels with confidential messages.}
    \label{fig:Model}
\end{figure}

We consider the quasi-static frequency-selective fading BCC illustrated in Fig. \ref{fig:Model}. The received signal $\yv[t],\zv[t]\in\CC^{N\times 1}$ of receiver 1, 2 at block $t$ is given by
\begin{eqnarray}\label{Model}
    \yv[t] &=& \Tc(\hv)\xv[t] + \nv[t]\\  \nonumber
    \zv[t] &=& \Tc(\gv) \xv[t]+ \nuv[t] ,\;\;\; t=1,\dots,n
\end{eqnarray}
where $\Tc(\hv),\Tc(\gv)$ denotes a $N \times (N+L)$ Toeplitz matrix with the $L+1$-path channel vector $\hv=[h_L,\dots,h_0]$ of user 1, $\gv=[g_L,\dots,g_0]$ of user 2, respectively, $\xv[t]\in\CC^{(N+L)\times 1}$ denotes the transmit
vector, and finally $\nv[t],\nuv[t]
\sim\Nc_{\Cc}(\zerov, \Id_N)$ are mutually independent additive white Gaussian noise (AWGN). The input vector
is subject to the power constraint given by
\begin{equation}\label{PConstraintX}
 \frac{1}{n} \sum_{t=1}^n \xv[t]^H \xv[t]\leq \overline{P}
\end{equation}
where we let $\overline{P}=(N+L) P$. The structure of $\Tc(\hv)$ is given by
\begin{small}
\begin{eqnarray*}
\Tc(\hv) =
\left[ \begin{array}{cccccc}
h_{L} & \cdots & h_0 & 0 & \cdots & 0 \\
0 & \ddots &  & \ddots & \ddots & \vdots \\
\vdots & \ddots & \ddots &  & \ddots & 0 \\
0 & \cdots & 0 & h_{L} & \cdots & h_0 \\
\end{array} \right].
\end{eqnarray*}
\end{small}
We assume that the channel matrices $\Tc(\hv), \Tc(\gv)$ remain constant for the whole duration of the transmission of $n$ blocks and are known to all terminals.
At each block $t$, we transmit $N+L$ symbols by appending a guard interval of size
$L\ll N$ larger than the delay spread, which enables to avoid the
interference between neighbor blocks.

The transmitter wishes to send a common message message $W_0$ to two receivers and a confidential message $W_1$ to receiver 1.
A $(2^{nR_0}, 2^{nR_1}, n)$ code consists of the following : 1) two message sets $\Wc_0=\{1,\dots, 2^{n R_0}\}$ and $\Wc_1=\{1,\dots, 2^{n R_1}\}$ with the messages $W_0,W_1$ uniformly distributed over the sets $\Wc_0,\Wc_1$, respectively ; 2)
a stochastic encoder that maps each message pair $(w_0,w_1)\in (\Wc_0,\Wc_1)$ to a codeword $\xv^n$ ; 3) One decoder at receiver 1 that maps a received sequence $\yv^n$ to a message pair $(\hat{w}_0^{(1)},\hat{w}_1)\in(\Wc_0,\Wc_1)$ and another at receiver 2 that maps a received sequence $\zv^n$ to a message $\hat{w}_0^{(2)}\in\Wc_0$. The average error probability of a $(2^{nR_0}, 2^{nR_1}, n)$ code is defined as
\begin{equation}
    P_e^{n} = \frac{1}{2^{n R_0}2^{n R_1}}\sum_{w_0\in \Wc_0}\sum_{w_1\in\Wc_1} P_e^n(w_0,w_1)
\end{equation}
where $P_e^n(w_0,w_1)$ denotes the error probability when the message pair $(w_0,w_1)$ is sent defined by
\begin{equation}
    P_e^n(w_0,w_1) \eqdef \Pr\left((\hat{w}_0^{(1)},\hat{w}_1) \neq (w_0,w_1) \cup \hat{w}_0^{(2)} \neq w_0 \right)
\end{equation}

The secrecy level of the confidential message $W_1$ at receiver 2 is measured by the equivocation rate $R_e$ defined as
\begin{equation}
    R_e \eqdef \frac{1}{n} H(W_1|Z^n)
\end{equation}
which is the normalized entropy of the confidential message conditioned on the received signal at receiver 2 and available CSI.

A rate-equivocation tuple $(R_0,R_1,R_e)$ is said to be achievable if for any $\epsilon>0$ there exists a sequence of codes $(2^{nR_0}, 2^{nR_1}, n)$ such that, we have
\begin{eqnarray}\nonumber
    P_e^n &\leq & \epsilon \\ \label{SecrecyNotion}
     R_1 -R_e & \leq & \epsilon
\end{eqnarray}
In this paper, we focus on the perfect secrecy case where receiver 2 obtains no information about the confidential message $W_1$, which is equivalent to $R_e=R_1$. In this setting,
an achievable rate region $(R_0,R_1)$ of the general BCC (expressed in bit per channel use per dimension) is given by \cite{csiszar1978bcc}
\begin{eqnarray}\label{Csiszar}
\Cc_s = \bigcup_{p(u,v,\xv)} \left\{ (R_0,R_1) : R_0 \leq \frac{1}{N+L}\min \{I(U;Y),I(U;Z)\},
R_1 \leq \frac{1}{N+L}[I(V;Y|U)-I(V;Z|U)]
\right\}
\end{eqnarray}
where the union is over all possible distribution $U,V,X$ satisfying \cite[Lemma 1]{choo2008krb}
\begin{equation}
U, V\rightarrow X \rightarrow Y, Z
\end{equation}
where $U$ might be a deterministic function of $V$. Recently, the secrecy capacity region $\Cc_s$ of the two-user MIMO-BCC (\ref{Model}) was characterized in \cite{Liangisita08} and is given by all possible rate tuples $(R_0,R_1)$ satisfying
\begin{eqnarray}\label{MIMOBCC}
R_0 &\leq & \frac{1}{N+L}\min \left\{\log\frac{|\Id + \Hm\Sm\Hm^H|}{|\Id +  \Hm\Km\Hm^H| },\log\frac{|\Id + \Gm\Sm\Gm^H|}{|\Id +  \Gm\Km\Gm^H| }\right\}\\
R_1 & \leq &\frac{1}{N+L}[\log|\Id +  \Hm\Km\Hm^H |-\log|\Id +  \Gm\Km\Gm^H |]
\end{eqnarray}
for some $\zerov \preceq \Km\preceq  \Sm$ with $\Sm$ denotes the input covariance satisfying $\trace(\Sm)\leq \overline{P}$ and $\Hm,\Gm$ denotes the channel matrix of receiver 1, 2, respectively.
Obviously, when only the confidential message is transmitted to receiver 1,
the frequency-selective BCC (\ref{Model}) reduces to the MIMO flat-fading wiretap channel whose secrecy capacity has been characterized in \cite{khisti:gmw,khisti2007stm,oggier2007scm,liu2007nsc,immse09}. In particular, Bustin et al. derived its closed-form expression under a power-covariance constraint \cite{immse09}. Under a total power (trace) constraint, the secrecy capacity of the MIMO Gaussian wiretap channel is expressed as \cite[Theorem 3]{LiuLiuPoorShamaiISIT09} \footnote{In \cite{immse09,LiuLiuPoorShamaiISIT09} the authors consider the real matrices $\Hm,\Gm$. Nevertheless, it is conjectured that for complex matrices the following expression without $1/2$ in the pre-log holds. }
\begin{equation}\label{MIMOMEcapacity}
  C_s= \frac{1}{N+L} \bigcup_{\Sm\succeq \zerov ; \trace(\Sm)\leq \overline{P}} \sum_{j=1}^r \log \phi_j
\end{equation}
where $\{\phi_j\}_{j=1}^{r}$ are the generalized eigen-values greater than one of the following pencil
\begin{equation}
    \left(\Id + \Sm^{1/2}\Hm\Hm^H\Sm^{1/2} , \Id + \Sm^{1/2}\Gm\Gm^H\Sm^{1/2}\right).
\end{equation}
As explicitly characterized in \cite[Theorem 2]{immse09}, the optimal input covariance achieving the above region is chosen such that the confidential message is sent over $r$ sub-channels where receiver 1 observes stronger signals than receiver 2. Moreover, in the high SNR regime the optimal strategy converges to beamforming into the null subspace of $\Gm$ \cite{khisti:gmw,khisti2007mc,liu2007nsc,oggier2007scm} as for the MISO case \cite{khisti:gmw,LiuPoorIT09}. In order to characterize the behavior of the secrecy capacity region in the high SNR regime, we define the d.o.f. region as
\begin{eqnarray}\label{def:sdof}
    (r_0,r_1)\eqdef \lim_{P\rightarrow \infty}\left(\frac{R_0}{\log P},\frac{R_1}{\log P}\right)
\end{eqnarray}
where $r_1$ denotes s.d.o.f. which corresponds precisely to the number $r$ of the generalized eigen-values greater than one in the high SNR.

\section{Vandermonde precoding}\label{sec:vandermonde}

For the frequency-selective BCC specified in Section \ref{sec:model}, we wish to design a practical linear precoding scheme which
fully exploits the d.o.f. offered by the frequency-selective channel. We remarked previously that for a special case when only the confidential message is sent to receiver 1 (without a common message), the optimal strategy consists of beamforming the confidential signal into the null subspace of receiver 2. By applying this intuitive result to the special Toeplitz MIMO channels $\Tc(\hv),\Tc(\gv)$ while including a common message, we propose a linear precoding strategy named \emph{Vandermonde precoding}.
Prior to the definition of the Vandermonde precoding, we provide some properties of a Vandermonde matrix \cite{ryandebbah}.
\begin{property}\label{def-Vand}
Given a full-rank Toeplitz matrix $\Tc(\gv)\in\CC^{N\times (N+L)}$, there exists a Vandermonde matrix $\tilde{\Vm}_1\in\CC^{(N+L)\times l}$ for $l\leq L$ whose structure is given by
\begin{eqnarray}\label{Vandermonde}
\tilde{\Vm}_1 = \left[ \begin{array}{cccccc}
1 & \cdots & 1 \\
a_1 & \cdots & a_{l} \\
a^2_1 & \cdots & a^2_{l} \\
\vdots & \ddots & \vdots \\
a^{N+L-1}_1 & \cdots & a^{N+L-1}_{l} \\
\end{array} \right] 
\end{eqnarray}
where $\{a_1,\dots,a_l\}$ are the $l\leq L$ roots of the polynomial $S(z)=\sum_{i=0}^L g_i z^{L-i}$ with $L+1$ coefficients of the channel $\gv$. Clearly $\tilde{\Vm}_1$ satisfies the following orthogonal condition
\begin{equation}
 \Tc(\gv)\tilde{\Vm}_1 = \zerov_{N\times l}
\end{equation}
and $\rank(\tilde{\Vm}_1)=l$ if $a_1,a_2,\dots,a_l$ are all different.
\end{property}

It is well-known that as the dimension of $N$ and $L$ increases, the Vandermonde matrix $\tilde{\Vm}_1$ becomes ill-conditioned unless the roots are on the unit circle. In other words, the elements of each column either grow in energy or tend to zero \cite{ryandebbah}. Hence, instead of the brut Vandermonde matrix (\ref{Vandermonde}), we consider a unitary Vandermonde matrix obtained either by applying the
Gram-Schmidt orthogonalization or singular value decomposition (SVD) on $\Tc(\gv)$.

\begin{definition}\label{def-uniVand}
We let $\Vm_1$ be a unitary Vandermonde matrix obtained by orthogonalizing the columns of $\tilde{\Vm}_1$. We let $\Vm_0\in\CC^{(N+L)\times (N+L-l)}$ be a unitary matrix in the null space of $\Vm_1$ such that $\Vm_0^H \Vm_1=\zerov$. The common message $W_0$, the confidential message $W_1$, is sent along $\Vm_0, \Vm_1$, respectively.
We call $\Vm=[\Vm_0, \Vm_1]\in \CC^{(N+L)\times (N+L)}$ \emph{Vandermonde precoder}.
\end{definition}
Further, the precoding matrix $\Vm_1$ for the confidential message satisfies the following property.
\begin{lemma}\label{lemma:rank}
Given two Toeplitz matrices $\Tc(\hv),\Tc(\gv)$ where $\hv,\gv$ are linearly independent, there exists a unitary Vandermonde matrix $\Vm_1\in\CC^{(N+L)\times l}$ for $0\leq l \leq L$ satisfying
\begin{eqnarray} \label{orthogonality}
&&  \Tc(\gv) \Vm_1=\zerov_{N\times l}, \\
&&  \rank(\Tc(\hv)\Vm_1) =l.
\end{eqnarray}
\end{lemma}

\begin{proof}
Appendix \ref{appendix:rank-bcc1}.
\end{proof}

In order to send the confidential message intended to receiver 1 as well as the common message to both receivers over the frequency-selective channel (\ref{Model}), we consider the Gaussian superposition coding based on the Vandermonde precoder of Definition \ref{def-uniVand}. Namely, at block $t$, we form the transmit vector as
\begin{equation}
\xv[t] = {\Vm}_0\uv_0[t] +{\Vm}_1 \uv_1[t]
\end{equation}
where the common message vector $\uv_0[t]$ and the confidential message vector $\uv_1[t]$ are mutually independent Gaussian vectors with zero mean and covariance $\Sm_0,\Sm_1$, respectively.
Under this condition, the input covariances subject to
\begin{eqnarray}\label{CovConstraint}
    \trace(\Sm_0)+\trace(\Sm_1) \leq \overline{P}
\end{eqnarray}
satisfy the power constraint (\ref{PConstraintX}). We let $\Fc$ denote the feasible set $(\Sm_0,\Sm_1)$ satisfying (\ref{CovConstraint}).

\begin{theorem} \label{theorem:region}
The Vandermonde precoding achieves the following secrecy rate region
\begin{eqnarray}\label{Vandregion}
\Rc_s =\cov \bigcup_{(\Sm_0,\Sm_1)\in \Fc} \left\{
\begin{array}
[c]{l}%
(R_0,R_1) :\\
 R_0 \leq \frac{1}{N+L}\min \{\log\frac{|\Id_N + \Hm_0\Sm_0 \Hm_0^H + \Hm_1\Sm_1\Hm_1^H|}{|\Id_N +\Hm_1\Sm_1\Hm_1^H|}, \log\left|\Id_N + \Gm_0 \Sm_0\Gm_0^H\right|\}, \\
 R_1 \leq \frac{1}{N+L}\log\left|\Id_N + \Hm_1 \Sm_1\Hm_1^H\right|
\end{array}
\right\}
\end{eqnarray}
where $\cov$ denotes the convex-hull and we let $\Hm_0= \Tc(\hv){\Vm}_0,\Hm_1=\Tc(\hv){\Vm}_1,\Gm_0= \Tc(\gv) {\Vm}_0 $.
\end{theorem}

\begin{proof}
Due to the orthogonal property (\ref{orthogonality}) of the unitary Vandermonde matrix, receiver 2 only observes the common message, which yields the received signals given by
\begin{eqnarray}
 \yv &=& \Tc(\hv) \Vm_0 \uv_0 + \Tc(\hv) \Vm_1 \uv_1 + \nv\\ \nonumber
 \zv &=& \Tc(\gv) \Vm_0 \uv_0 + \nuv
\end{eqnarray}
where we drop the block index.
We examine the achievable rate region $\Rc^{\rm}_s$ of the Vandermonde precoding. By letting
the auxiliary variables $U= \Vm_0 \uv_0, V = U + \Vm_1\uv_1$ and $X=V$,
we have
\begin{eqnarray*}\label{I(u;y)}
 I(U;Y)&=& \frac{1}{N+L}\log\frac{|\Id_N + \Tc(\hv){\Vm}_0\Sm_0 {\Vm}_0^H\Tc(\hv)^H + \Tc(\hv){\Vm}_1\Sm_1{\Vm}_1^H\Tc(\hv)^H|}
{|\Id_N +\Tc(\hv){\Vm}_1\Sm_1 {\Vm}_1^H\Tc(\hv)^H|}\\
I(U;Z) &=&  \frac{1}{N+L}\log\left|\Id_N +  \Tc(\gv) {\Vm}_0 \Sm_0{\Vm}_0^H \Tc(\gv)^H\right|\\  \label{Iprivate}
I(V;Y|U) &=&  \frac{1}{N+L}\log\left|\Id_N + \Tc(\hv){\Vm}_1 \Sm_1 {\Vm}_1^H \Tc(\hv)^H\right|\\
I(V;Z|U) &=&  0
\end{eqnarray*}
Plugging these expressions to (\ref{Csiszar}), we obtain (\ref{Vandregion}).
\end{proof}

\begin{figure}[t]
    \centering
    \includegraphics[width=0.4\columnwidth]{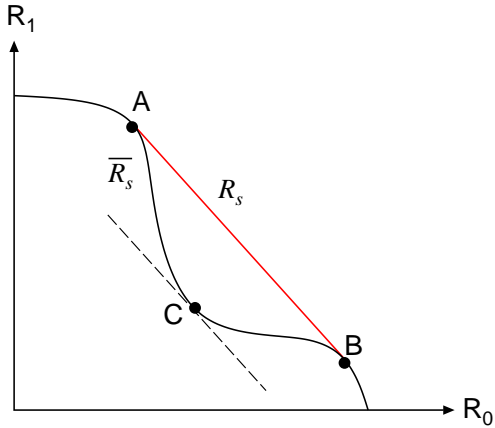}\\
    \vspace{-1em}
    \caption{Achievable rate region $\Rc_s$ obtained by the convex-hull on $\overline{\Rc}_s$.}
    \label{fig:region}
\end{figure}
The boundary of the achievable rate region of the Vandermonde precoding can be characterized by solving the weighted sum rate maximization.
Any point $(R_0^{\star},R_1^{\star})$ on the boundary of the convex region $\Rc_s$ is obtained by solving
\begin{equation}\label{WSRM}
\max_{(R_0,R_1)\in \Rc_s} \gamma_0 R_0 + \gamma_1 R_1
\end{equation}
for non-negative weights $\gamma_0,\gamma_1$ satisfying $\gamma_0+\gamma_1=1$. When the region $\overline{\Rc}_s$, obtained without convex hull,  is non-convex, the set of the optimal covariances $(\Sm_0^{\star},\Sm_1^{\star})$ achieving the boundary point might not be unique. Fig. \ref{fig:region} depicts an example in which the achievable rate region $\Rc_s$ is obtained by the convex hull operation on the region $\overline{\Rc}_s$, i.e. replacing the non-convex subregion by the line segment $A,B$. For the weight ratio $\gamma_1/\gamma_0$ corresponding to the slope of the line segment $A,B$, there exist two optimal sets of the covariances yielding the points A and B (which clearly dominate the point C). These points are the solution to the weighted sum rate maximization (\ref{WSRM}). In summary, an optimal covariance set achieving (\ref{WSRM}) (might not be unique) is the solution of
\begin{equation}\label{WSRM2}
\max_{(\Sm_0,\Sm_1)\in \Fc} \gamma_0 R_0 + \gamma_1 R_1=
 \max_{(\Sm_0,\Sm_1)\in \Fc} \gamma_0 \min\{R_{01},R_{02}\} + \gamma_1 R_1
\end{equation}
where we let
\begin{eqnarray*}
R_{01}(\Sm_0,\Sm_1)&=& \frac{1}{N+L}\log\frac{|\Id_N +\Hm_0\Sm_0\Hm_0^H +\Hm_1\Sm_1\Hm_1^H|}{|\Id_N +\Hm_1\Sm_1\Hm_1^H|}\\
R_{02}(\Sm_0) &=&  \frac{1}{N+L}\log\left|\Id_N +\Gm_0\Sm_0 \Gm_0^H\right| \\
R_1(\Sm_1) &=& \frac{1}{N+L}\log\left|\Id_N +\Hm_1\Sm_1\Hm_1^H\right|
\end{eqnarray*}
Following \cite[Section II-C]{poor1994isd} (and also \cite[Lemma 2]{liang2007sco}), we remark that the solution to the max-min problem (\ref{WSRM2}) can be found by hypothesis testing of three cases, $R_{01}<R_{02}, R_{02}<R_{01}$, and $R_{01}=R_{02}$. Formally, we have the following lemma.
\begin{lemma} \label{maxmin}
The optimal $(\Sm_0^{\star},\Sm_1^{\star})$, solution of (\ref{WSRM2}), is given by one of the three solutions.
\begin{enumerate}
   \item Case 1 : $(\Sm_0^{\star},\Sm_1^{\star})$ maximizes
      \begin{equation}\label{f1}
      f_1(\Sm_0,\Sm_1)= \gamma_0\log\frac{|\Id_N + \Hm_1\Sm_1\Hm_1^H + \Hm_0\Sm_0\Hm_0^H|}{|\Id_N +\Hm_1\Sm_1\Hm_1^H |} + \gamma_1 \log\left|\Id_N +\Hm_1\Sm_1\Hm_1^H\right|
      \end{equation}
      and satisfies $R_{01}(\Sm_0^{\star},\Sm_1^{\star})< R_{02}(\Sm_1^{\star})$.
   \item Case 2 : $(\Sm_0^{\star},\Sm_1^{\star})$ maximizes
      \begin{equation}\label{f2}
      f_2(\Sm_0,\Sm_1)= \gamma_0  \log\left|\Id_N +\Gm_0\Sm_0\Gm_0^H\right| + \gamma_1 \log\left|\Id_N +\Hm_1\Sm_1\Hm_1^H\right|
      \end{equation}
      and satisfies $R_{02}(\Sm_1^{\star})< R_{01}(\Sm_0^{\star},\Sm_1^{\star})$.
  \item Case 3 : $(\Sm_0^{\star},\Sm_1^{\star})$ maximizes
      \begin{eqnarray}\label{f3}
      f_3(\Sm_0,\Sm_1) &=& \gamma_0 \left[\theta \log\frac{|\Id_N + \Hm_1\Sm_1\Hm_1^H + \Hm_0\Sm_0\Hm_0^H|}{|\Id_N +\Hm_1\Sm_1\Hm_1^H |}
      +(1-\theta)\log\left|\Id_N +\Gm_0\Sm_0\Gm_0^H\right| \right] \\ \nonumber
       && +\gamma_1 \log\left|\Id_N +\Hm_1\Sm_1\Hm_1^H\right|
      \end{eqnarray}
      and satisfies $R_{01}(\Sm_0^{\star},\Sm_1^{\star})=R_{02}(\Sm_1^{\star})$ for some $0<\theta<1$.
\end{enumerate}
\end{lemma}

Before considering the weighted sum rate maximization (\ref{WSRM2}), we apply SVD to $\Hm_1\in\CC^{N\times l},\Gm_0\in\CC^{N\times (N+L-l)}$
\begin{eqnarray}\label{SVD}
\Hm_1&=& \Um_{h1} \Lambdam_{h1} {\Vm_{h1}}^H\\ \nonumber
\Gm_0 &=& \Um_{g0} \Lambdam_{g0} {\Vm_{g0}}^H
\end{eqnarray}
where $ \Um_{h1}, \Um_{g0}\in\CC^{N\times N}$, $\Vm_{h1}\in\CC^{l\times l}$, and $\Vm_{g0}\in\CC^{(N+L-l)\times (N+L-l)}$ are unitary,  $\Lambdam_{h1},\Lambdam_{g0}$ contain positive singular values $\{\sqrt{\lambda^{h1}_{i}}\}_{i=1}^{l}$, $\{\sqrt{\lambda^{g0}_i}\}_{i=1}^{N+L-l}$, respectively.
Following \cite[Theorem 3]{liang2007sco}, we apply Lemma \ref{maxmin} to solve the weighted sum rate maximization.

\begin{theorem}\label{theorem:covariance}
The set of the optimal covariances $(\Sm_0^{\star},\Sm_1^{\star})$, achieving the boundary of the achievable rate region $\Rc_s$ of the Vandermonde precoding, corresponds to one of the following three solutions.
\begin{enumerate}
\item Case 1 :$(\Sm_0^{\star},\Sm_1^{\star})=(\Sm_0^{1},\Sm_1^{1})$, if $(\Sm_0^{1},\Sm_1^{1})$, solution of the following KKT conditions,  satisfies $R_{01}(\Sm_0^{1},\Sm_1^{1})< R_{02}(\Sm_1^{1})$.
    \begin{eqnarray} \label{case1-1}
&&    \gamma_0  \Hm_0^H \Gammam^{-1}\Hm_0 + \Psim_0 = \mu\Id_{N+L-l} \\  \label{case1-2}
&&\gamma_0\Hm_1^H \Gammam^{-1}\Hm_1 +
(\gamma_1-\gamma_0)\Hm_1^H(\Id_N +\Hm_1\Sm_1\Hm_1^H)^{-1}\Hm_1 + \Psim_1 = \mu\Id_l
\end{eqnarray}
    where $\trace(\Psim_i\Sm_i)=0$ with a positive semidefinite $\Psim_i$ for $i=0,1$, $\mu\geq 0$ is determined such that $\trace(\Sm_0)+\trace(\Sm_1)=\overline{P}$, and we let $\Gammam=\Id_N +\Hm_0\Sm_0\Hm_0^H+ \Hm_1\Sm_1\Hm_1^H$.
  \item Case 2 :$(\Sm_0^{\star},\Sm_1^{\star})=(\Sm_0^{2},\Sm_1^{2})$ if the following $(\Sm_0^{2},\Sm_1^{2})$ fulfills $R_{02}(\Sm_1^{2})< R_{01}(\Sm_0^{2},\Sm_1^{2})$.\\
   We let $\Sm_0^2=\Vm_{g0}\hat{\Sm}_0\Vm_{g0}^H$ and $\Sm_1^2=\Vm_{h1}\hat{\Sm}_1\Vm_{h1}^H$
 where 
 $\hat{\Sm}_0,\hat{\Sm}_1$ are diagonal with the $i$-th element given by
 \begin{eqnarray}\label{solution2}
    p_{0,i} &=& \left[\frac{\gamma_0}{\mu}- \frac{1}{\lambda^{g0}_{i}}\right]_+ , \;\;i=1,\dots,N+L-l\\ \nonumber
    p_{1,i} &=& \left[\frac{\gamma_1}{\mu}- \frac{1}{\lambda^{h1}_{i}}\right]_+,\;\; i=1,\dots,l
 \end{eqnarray}
 where $\mu\geq 0$ is determined such that $\sum_{i=1}^{N+L-l}p_{0,i}+ \sum_{i=1}^l p_{1,i}=\overline{P}$.
  \item Case 3 : $(\Sm_0^{\star},\Sm_1^{\star})=(\Sm_0^{3},\Sm_1^{3})$, if $(\Sm_0^{3},\Sm_1^{3})$, solution of the following KKT conditions, satisfies $R_{02}^{\theta}(\Sm_1^{3})= R_{01}^{\theta}(\Sm_0^{3},\Sm_1^{3})$ for some $0< \theta<1$.
      \begin{eqnarray}\label{case3-1}
&      \theta\Hm_0^H \Gammam^{-1}\Hm_0+ (1-\theta)
\Gm_0^H\left(\Id_N + \Gm_0\Sm_0\Gm_0^H\right)^{-1}\Gm_0  + \Psim_0 = \mu\Id_{N+L-l} \\ \label{case3-2}
&\gamma_0 \theta \Hm_1^H \Gammam^{-1}\Hm_1 +
(\gamma_1-\gamma_0\theta)\Hm_1^H\left(\Id_N+\Hm_1\Sm_1\Hm_1^H\right)^{-1}\Hm_1 + \Psim_1 =\mu\Id_l
\end{eqnarray}
where $\trace(\Psim_i\Sm_i)=0$ with a positive semidefinite $\Psim_i$ for $i=0,1$,
$\mu\geq 0$ is determined such that $\trace(\Sm_0)+\trace(\Sm_1)=\overline{P}$.
\end{enumerate}
\end{theorem}

\begin{proof}
    Appendix \ref{appendix:covariance}.
\end{proof}

\begin{remark}
Due to the non-concavity of the underlying weighted sum rate functions, it is generally difficult to characterize the boundary of the achievable rate region $\Rc_s$ except for some special cases. The special cases include the corner points, in particular, the secrecy rate for the case of sending only the confidential message ($\gamma_1=1$), as well as the maximum sum rate point for the equal weight case ($\gamma_0=\gamma_1$). It is worth noticing that under equal weight the objective functions in three cases are all concave in $\Sm_0,\Sm_1$ since $f_1$ is concave if $\gamma_1\geq \gamma_0$ and $f_3$ is concave if $\gamma_1\geq \gamma_0\theta$ and $0<\theta<1$.
\end{remark}

The maximum sum rate point $\gamma_0=\gamma_1$ can be found by applying the following greedy search \cite{liang2007sco}.
\textbf{Greedy search to find the maximum sum rate point}
\begin{enumerate}
  \item Find $\Sm_0,\Sm_1$ maximizing $f_1$ and check $R_{02}<R_{01}$. If yes stop. Otherwise go to 2).
  \item Find $\Sm_0,\Sm_1$ maximizing $f_2$ and check $R_{01}<R_{01}$. If yes stop. Otherwise go to 3).
  \item Find $\Sm_0,\Sm_1$ maximizing $f_3$ and check $R_{01}^{\theta}=R_{01}^{\theta}$ for some $0<\theta<1$.
\end{enumerate}

For the special case of $\gamma_1=1$, Theorem \ref{theorem:covariance} yields the achievable secrecy rate with the Vandermonde precoding.
\begin{cor}
The Vandermonde precoding achieves the secrecy rate
\begin{eqnarray}\label{secrecyrate}
R_1^{\rm vdm} &=& \max_{\Sm_1: \trace(\Sm_1)\leq \overline{P}} \frac{1}{N+L}\log\det\left( \Id_{N} + \Tc(\hv)\Vm_1 \Sm_1 \Vm_1^H  \Tc(\hv)^H\right) \\ \nonumber
&=&  \frac{1}{N+L}\sum_{i=1}^L \log(\mu\lambda^{h1}_i)_+
\end{eqnarray}
where the last equality is obtained by applying SVD to $\Hm_1=\Tc(\hv)\Vm_1$ and plugging the power allocation of (\ref{solution2}) with $\gamma_0=0,\gamma_1=1$, $\mu$ is determined such that $\sum_{l=1}^L p_{1i}\leq \overline{P}$.
\end{cor}

Finally, by focusing the behavior of the achievable rate region in the high SNR regime, we characterize the achievable d.o.f. region of the frequency-selective BCC (\ref{Model}).
\begin{theorem}\label{theorem:DoF-2user}
The d.o.f. region of the frequency-selective BCC (\ref{Model}) with $(N+L)\times L$ Toeplitz matrices $\Tc(\hv),\Tc(\gv)$ is given as a union of $(r_0,r_1)= \frac{1}{N+L}(l_0,l)$ satisfying
\begin{eqnarray} \label{MIMOsdof}
&& l \leq L \\  \label{indMIMOlink}
&& l_0 + l \leq N
\end{eqnarray}
where $l_0,l$ denote non-negative integers. The Vandermonde precoding achieves the above d.o.f. region.
\end{theorem}

\begin{figure}[t]
    \centering
    \includegraphics[width=0.4\columnwidth]{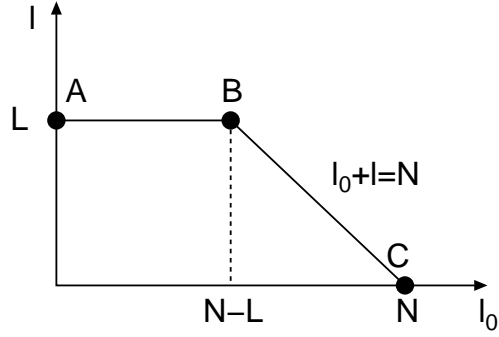}\\
    \vspace{-1em}
    \caption{d.o.f. region $(l_0,l_1)$ of frequency-selective BCC.}
    \label{fig:2userDoF}
\end{figure}
\begin{proof}
The achievability follows rather trivially by applying Theorem \ref{theorem:region}.
By considering equal power allocation over all $N+L$ streams such that $\Sm_0=P\Id_{N+L-l}, \Sm_1=P\Id_l$, we obtain the rate tuple $(R_0,R_1)$ where $R_0\leq \min(R_{01},R_{02})$
\begin{eqnarray*}
R_{01}&=& \frac{1}{N+L}\log\frac{|\Id_N + P \Tc(\hv) \Vm_0\Vm_0^H\Tc(\hv)^H + P\Tc(\hv) \Vm_1 \Vm_1^H \Tc(\hv)^H|}
{|\Id_N +P\Tc(\hv) \Vm_1 \Vm_1^H \Tc(\hv)^H|}\\
 R_{02}&=&  \frac{1}{N+L}\log\left|\Id_N +P\Tc(\gv) \Vm_0 \Vm_0^H\Tc(\gv)^H\right|\\  \label{EqualPowerSecRate}
R_1 &\leq &  \frac{1}{N+L}\log\left|\Id_{N} + P\Tc(\hv) \Vm_1\Vm_1^H\Tc(\hv)^H\right|
\end{eqnarray*}
We first notice that the pre-log factor of $\log|\Id + P\Am|$ as $P\rightarrow \infty$ depends only on the rank of $\Am$. From Lemma \ref{lemma:rank}, we obtain
\begin{eqnarray}\label{R1}
\rank(\Tc(\hv) \Vm_1\Vm_1^H\Tc(\hv)^H) &=& \rank(\Tc(\hv) \Vm_1)= l \\ \nonumber
\rank(\Tc(\gv) \Vm_0 \Vm_0^H\Tc(\gv)^H) &=& \rank(\Tc(\gv)\Vm_0) \\ \nonumber
&\overset{\mathrm{(a)}}=& \rank(\Tc(\gv)[\Vm_0 \Vm_1]) \\ \label{R02}
&\overset{\mathrm{(b)}}=& \rank(\Tc(\gv))=N\\ \nonumber
\rank(\Tc(\hv) (\Vm_0\Vm_0^H+\Vm_1\Vm_1^H)\Tc(\hv)^H) &=& \rank\left(\Tc(\hv)\Vm\Vm^H\Tc(\hv)^H\right)\\ \label{R01}
         &\overset{\mathrm{(b)}} =& \rank(\Tc(\hv))=N
\end{eqnarray}
where (a) follows from orthogonality between $\Tc(\gv)$ and $\Vm_1$, (b)
follows from the fact that $\Vm=[\Vm_0\Vm_1]$ is unitary satisfying $
\Vm\Vm^H=\Id$. Notice that (\ref{R1}) yields $r_1=\frac{l}{N+L}$. For the d.o.f. $r_0=\frac{l_0}{N+L}$ of the common message, (\ref{R1}) and (\ref{R01}) yield
\begin{equation}
l_0= \rank(\Tc(\hv) (\Vm_0\Vm_0^H+\Vm_1\Vm_1^H)\Tc(\hv)^H)-\rank(\Tc(\hv) \Vm_1\Vm_1^H\Tc(\hv)^H)=N-l
\end{equation}
which is dominated by the pre-log of $R_{02}$ in (\ref{R02}). This establishes the achievability.

The converse follows by noticing that the inequalities (\ref{MIMOsdof}) and (\ref{indMIMOlink}) correspond to trivial upper bounds.
The first inequality (\ref{MIMOsdof}) corresponds to the s.d.o.f. of the MIMO wiretap channel with the legitimate channel $\Tc(\hv)$ and the eavesdropper channel $\Tc(\gv)$, which is bounded by $L$. The second inequality (\ref{indMIMOlink}) follows because the total number of streams for receiver 1 cannot be larger than the d.o.f. of $\Tc(\hv)$, i.e. $N$.
\end{proof}

Fig. \ref{fig:2userDoF} illustrates the region $(l,l_0)$ of the frequency-selective BCC over $N+L$ dimensions. We notice that the s.d.o.f. constraint (\ref{MIMOsdof}) yields the line segment A, B while the constraint (\ref{indMIMOlink}) in terms of the total number of streams for receiver 1
yields the line segment B,C.

\section{Multi-user secure communications}\label{sec:multiuser}
In this section, we provide some applications of the Vandermonde precoding in the multi-user secure communication scenarios where the transmitter wishes to send confidential messages to more than one intended receivers. The scenarios that we address are ; a) a $K+1$-user frequency-selective BCC with $K$ confidential messages and one common message, b) a two-user frequency-selective BCC with two confidential messages and one common message.
For each scenario, by focusing on the behavior in the high SNR regime, we characterize the achievable s.d.o.f. region and show the optimality of the Vandermonde precoding.

\subsection{$K$+1-user BCC with $K$ confidential messages}\label{subsec:BCCouter}
As an extension of Section \ref{sec:vandermonde}, we consider the $K+1$-user frequency-selective BCC where the transmitter sends $K\leq L$ confidential messages $W_1,\dots,W_K$ to the first $K$ receivers as well as one common message $W_0$ to all receivers. Each of the confidential messages must be kept secret to receiver $K+1$.
Notice that this model, called multi-receiver wiretap channel, has been studied in the literature (\cite{bagherikaram2008srr,bagherikaram2009scr,ekrem2008scc,ekrem2009scr,choo2009trb,choo2008krb} and reference therein). In particular, the secrecy capacity region of the Gaussian MIMO multi-receiver wiretap channel has been characterized in \cite{bagherikaram2009scr,ekrem2009scr} for $K=2$, an arbitrary $K$, respectively, where the optimality of the S-DPC is proved.

The received signal $\yv_k$ of receiver $k$ and the received signal $\zv$ of receiver $K+1$ at any block are given by
\begin{eqnarray}\label{BCCwKusers}
\yv_k &=& \Tc(\hv_k) \xv + \nv_k, \;\; k=1,\dots,K \\
\zv &=& \Tc(\gv) \xv + \nuv
\end{eqnarray}
where $\xv$ is the transmit vector satisfying the total power constraint and $\nv_1,\dots,\nv_K,\nuv$ are mutually independent AWGN with covariance $\Id$.
We assume that the $K+1$ vectors $\hv_1,\dots,\hv_K,\gv$ of length $L+1$ are linearly independent and perfectly known to all the terminals.
As an extension of the frequency-selective BCC in Section \ref{sec:model}, we say that the rate tuple $(R_0,R_1,\dots,R_K)$ is achievable if for any $\epsilon>0$ there exists a sequence of codes $(2^{nR_0}, 2^{nR_1},\dots,2^{nR_K}, n)$
such that
\begin{eqnarray}
&& P_e^n\leq \epsilon \\ \label{SecrecyRatek}
&& \sum_{k\in \Kc} R_k -\frac{1}{n} H(W_{\Kc}|Z^n) \leq \epsilon, \;\; \Kc\subseteq \{1,\dots,K\}
\end{eqnarray}
where we denote $W_{\Kc}=\{\forall k\in \Kc, W_k\}$ and define
\begin{equation}
    P_e^n = \frac{1}{\prod_{k=0}^K 2^{nR_k}}\sum_{w_0\in \Wc_0}\dots\sum_{w_K\in \Wc_K} \Pr\left(\cup_{k=1}^K (\hat{w}_0^{(k)},\hat{w}_k) \neq (w_0,w_k)  \right).
\end{equation}
An achievable secrecy rate region $(R_1,R_2)$ for the case of $K=2$, when the transmitter sends two confidential messages in the presence of an external eavesdropper, is provided in \cite[Theorem 1]{bagherikaram2008srr}. This theorem can be extended to an arbitrary $K$ while including the common message. Formally we state the following lemma.
\begin{lemma}\label{lemma:region}
An achievable rate region of the $K$+1-user BCC, where the transmitter sends $K$ confidential messages intended to the first $K$ receivers as well as a common message to all users, is given as a union of all non-negative rate-tuple satisfying
\begin{eqnarray}\nonumber
        R_0 &\leq &\min\{I(U;Z), \min_{k} I(U;Y_k)\} \\ \nonumber
        R_k & \leq & I(V_k;Y_k|U) - I(V_k;Z|U), \;\;\; k=1,\dots,K \\  \nonumber
        \sum_{k\in\Kc} R_k &\leq & \sum_{k\in\Kc} I(V_k;Y_k|U) - \sum_{j=2}^{|\Kc|}I(V_{\pi(j)}; V_{\pi(1)},\dots,V_{\pi(j-1)}|U)  - I(V_{\Kc} ; Z|U),\\ \label{RegionKuserBCC}
        && \; \;\forall \Kc \subseteq \{1,\dots,K\}, \forall \pi
\end{eqnarray}
where $\pi$ denotes a permutation over the subset $\Kc$, $|\Kc|$ denotes the cardinality of $\Kc$, we let $V_{\Kc}=\{\forall k\in \Kc, V_k\}$,  and the random variables $U,V_1,\dots,V_K,X,Y_1,\dots,Y_K,Z$ satisfy the Markov chain
\begin{equation}\label{MarkovKuser}
 U,V_1,\dots,V_K\rightarrow X \rightarrow Y_1,\dots,Y_K,Z.
\end{equation}
\end{lemma}

\begin{proof}
    Appendix \ref{appendix:achievability}.
\end{proof}
Notice that the second term of the last equation in (\ref{RegionKuserBCC}) can be also expressed by
\begin{eqnarray}\nonumber
         \sum_{j=2}^{|\Kc|}I(V_{\pi(j)}; V_{\pi(1)},\dots,V_{\pi(j-1)}|U) = \sum_{k\in \Kc}H(V_k|U) - H(V_{\Kc}|U). \; \;\forall \Kc \subseteq \{1,\dots,K\}, \forall \pi
\end{eqnarray}
It can be easily seen that without the secrecy constraint the above region reduces to the Marton's achievable region for the general $K$-user broadcast channel \cite{viswanath2003scv}.

In order to focus on the behavior of the region in the high SNR regime, we define the s.d.o.f. region as
\[ r_0 = \lim_{P\rightarrow \infty}\frac{R_0}{\log P }, \;\; r_k = \lim_{P\rightarrow\infty}\frac{R_k}{\log P }, \;\;k=1,\dots,K \]
where $r_0$ denotes the d.o.f. of the common message and
$r_k$ denotes the s.d.o.f. of confidential message $k$. As an extension of Theorem \ref{theorem:DoF-2user}, we have the following s.d.o.f. region result.
\begin{theorem}\label{theorem:DoFouterBCC}
The s.d.o.f. region of the $K+1$-user frequency-selective BCC (\ref{BCCwKusers}) is a union of
$(r_0,r_1,\dots,r_K)=\frac{1}{N+L}(l_0,l_1,\dots,l_K)$ satisfying
\begin{eqnarray} \label{CooperativeBound1}
    && \sum_{k=1}^K l_k \leq L \\ \label{CooperativeBound2}
    &&  l_0 + \sum_{k=1}^K l_k \leq N
\end{eqnarray}
where $\{l_0,l_1,\dots,l_K\}$ are non-negative integers. The Vandermonde precoding achieves this region.
\end{theorem}

\begin{proof}
    Appendix \ref{appendix:SDOF1}.
\end{proof}

\begin{figure}[t]
    \centering
    \includegraphics[width=0.4\columnwidth]{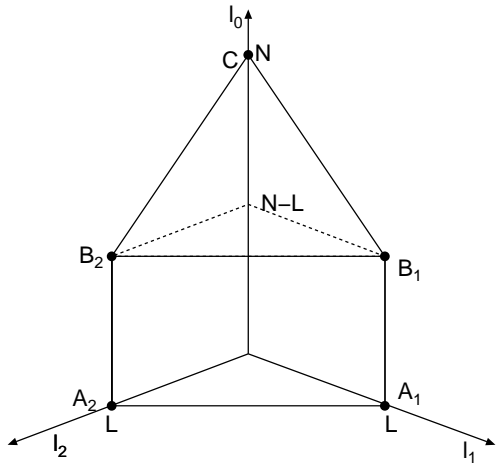}\\
    \vspace{-1em}
    \caption{s.d.o.f. region $(l_0,l_1,l_2)$ over $N+L$ dimensions of three-user frequency-selective BCC.}
    \label{fig:3userDoF}
\end{figure}

Fig. \ref{fig:3userDoF} illustrates the region $(l_0,l_1,l_2)$ for the case of $K=2$ confidential messages. It can be easily seen that the constraint (\ref{CooperativeBound2}) in terms of the total number of streams for the virtual receiver yields the subspace $C,B_1,B_2$ while the s.d.o.f. constraint (\ref{CooperativeBound1}) for the virtual receiver yields the subspace $A_1, A_2,B_2,B_1$. We remark that for the special case of one confidential message and one common message ($K=1$), the region reduces to Fig. \ref{fig:2userDoF}.

\begin{figure}[t]
    \centering
    \includegraphics[width=0.4\columnwidth]{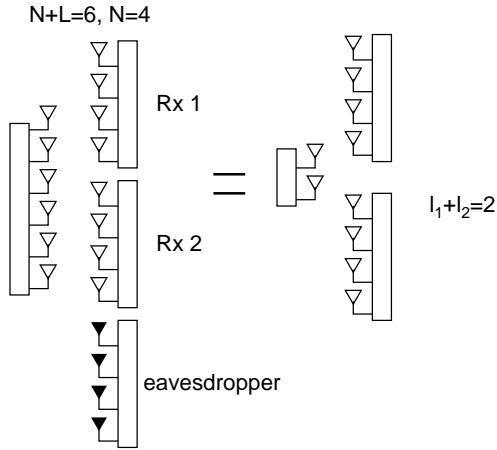}\\
    \vspace{-1em}
    \caption{Equivalent MIMO interpretation for three-user frequency-selective BCC with two confidential messages.}
    \label{fig:Interprete1}
\end{figure}

\begin{remark}
When only the $K$ confidential messages are transmitted to the $K$ intended receivers in the presence of the eavesdropper, the s.d.o.f. region has the equivalent MIMO interpretation \cite{lee2007hsa}. More specifically, the frequency-selective BCC (\ref{BCCwKusers}) is equivalent to the MIMO-BCC where the transmitter with $N+L$ dimensions (antennas) sends messages to $K$ receivers with $N$ antennas each in the presence of the eavesdropper with $N$ antennas. The secrecy constraint (orthogonal constraint) consumes $N$ dimensions of the channel seen by the virtual receiver and lets the number of effective transmit antennas be $L$. The resulting channel is the MIMO-BC without secrecy constraint with $L$ transmit antennas and $K$ receivers with $N$ antennas each, whose multiplexing gain is $\min(L,KN)=L$ (we assume $L<N$). Fig. \ref{fig:Interprete1} illustrates the example with $K=2,N=4,L=2$.
\end{remark}
\subsection{Two-user BCC with two confidential messages} \label{subsec:BCCinner}
We consider the two-user BCC where the transmitter sends two confidential messages $W_1,W_2$ as well as one common message $W_0$. Each of the confidential messages must be kept secret to the unintended receiver. This model has been studied in \cite{liu2008dmi,LiuPoorIT09,LiuLiuPoorShamaiISIT09} for the case of two confidential messages and in \cite{choo2008krb} for the case of two confidential messages and a common message. In \cite{LiuLiuPoorShamaiISIT09}, the secrecy capacity region of the MIMO Gaussian BCC was characterized.
The received signal at receiver 1, 2 at any block is given respectively by
\begin{eqnarray}\label{two-userBCC}
\yv_1 &=& \Tc(\hv_1) \xv + \nv_1  \\ \nonumber
\yv_2 &=& \Tc(\hv_2) \xv + \nv_2
\end{eqnarray}
where $\xv$ is the input vector satisfying the total power constraint and $\nv_1,\nv_2$ are mutually independent AWGN with covariance $\Id_N$. We assume the channel vectors $\hv_1,\hv_2$ are linearly independent.

We say that the rate tuple $(R_0,R_1,R_2)$ is achievable if for any $\epsilon>0$ there exists a sequence of codes $(2^{nR_0}, 2^{nR_1},2^{nR_2}, n)$
such that
\begin{eqnarray}
&&P_e^n\leq \epsilon \\ \label{SecrecyRateInnerk}
&&R_1 -\frac{1}{n} H(W_1|Y_2^n) \leq\epsilon , \;\; R_2 -\frac{1}{n} H(W_2|Y_1^n) \leq \epsilon
\end{eqnarray}
where we define the average error probability as
\begin{equation}
    P_e^n = \frac{1}{\prod_{k=0}^2 2^{nR_k}}\sum_{w_0\in \Wc_0}\sum_{w_1\in \Wc_1}\sum_{w_2\in \Wc_2} \Pr\left((\hat{w}_0^{(1)},\hat{w}_1) \neq (w_0,w_1) \cup (\hat{w}_0^{(2)},\hat{w}_2) \neq (w_0,w_2) \right)
\end{equation}
where $(\hat{w}_0^{(1)},\hat{w}_1), (\hat{w}_0^{(2)},\hat{w}_2)$ is the output of decoder 1, 2, respectively.
A secrecy achievable rate region of the two-user BCC with two confidential messages and a common message is given by \cite[Theorem 1]{choo2008krb}
\begin{eqnarray}\label{two-userBCCRate}
R_0 &\leq &\min\{ I(U;Y_1), I(U; Y_2)\} \\ \nonumber
R_1 &\leq & I(V_1;Y_1|U) - I(V_1;Y_2,V_2|U) \\ \nonumber
R_2 &\leq & I(V_2;Y_2|U) - I(V_2;Y_1,V_1|U)
\end{eqnarray}
where the random variables satisfy the Markov chain
\begin{equation}
 U, V_1,V_2 \rightarrow X \rightarrow Y_1, Y_2.
\end{equation}

We extend Theorem \ref{theorem:DoF-2user} to the two-user frequency-selective BCC (\ref{two-userBCC}) and obtain the following s.d.o.f. result.
\begin{theorem}\label{theorem:DoFinnerBC}
The s.d.o.f. region of the two-user frequency-selective BCC (\ref{two-userBCC}) is a union of $(r_0,r_1,r_2)=\frac{1}{N+L}(l_0,l_1,l_2)$ satisfying
\begin{eqnarray}\label{peruserSDoF}
    &&  l_k\leq L, \;\; k=1,2 \\ \label{MIMOconst}
    &&  l_0 + l_k \leq N, \;\; k=1,2
\end{eqnarray}
where $\{l_0,l_1,l_2\}$ are non-negative integers. The Vandermonde precoding achieves the region.
\end{theorem}

\begin{proof}
    Appendix \ref{appendix:SDOF2}.
\end{proof}

\begin{figure}[t]
    \centering
    \includegraphics[width=0.4\columnwidth]{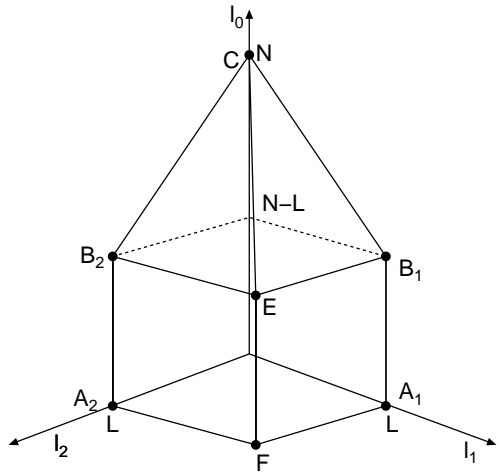}\\
    \vspace{-1em}
    \caption{s.d.o.f. region $(l_0,l_1,l_2)$ over $N+L$ dimensions of $K=2$-user frequency-selective BCC.}
    \label{fig:inner2userDoF}
\end{figure}
Fig. \ref{fig:inner2userDoF} represents the s.d.o.f. region $(l_0,l_1,l_2)$ over $N+L$ dimensions of the two-user frequency-selective BCC. The per-receiver s.d.o.f. constraints (\ref{peruserSDoF}) yield the subspace $A_1,B_1,E, F$ for user 1 and the subspace $A_2,B_2,E, F$ for user 2. The constraints (\ref{MIMOconst}) in terms of the total number of streams per receiver yield the subregion $C,B_1,E$ for user 1 and the subregion $C,B_2,E$ for user 2. For the special case of one confidential message and one common message, the region reduces to Fig. \ref{fig:2userDoF}.

\begin{remark}
Comparing Theorems \ref{theorem:DoFouterBCC},\ref{theorem:DoFinnerBC} as well as Figs. \ref{fig:3userDoF},\ref{fig:inner2userDoF} for $K=2$, it clearly appears that the s.d.o.f. of $K+1$-user BCC with $K$ confidential messages is dominated by the s.d.o.f. of $K$-user BCC with $K$ confidential messages. In other words, the s.d.o.f. region critically depends on the assumption on the eavesdropper(s) to whom each confidential message must be kept secret.
\end{remark}

\begin{figure}[t]
    \centering
    \includegraphics[width=0.4\columnwidth]{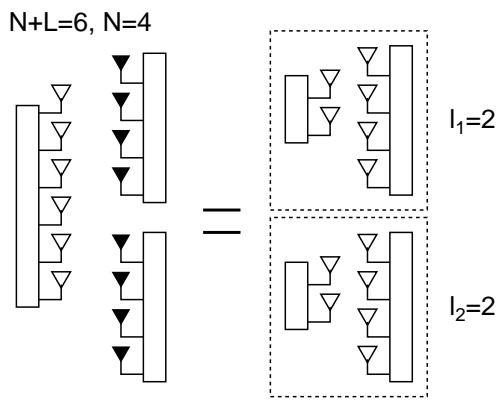}\\
    \vspace{-1em}
    \caption{Equivalent MIMO interpretation for the two-user frequency-selective BCC with two confidential messages.}
    \label{fig:Interprete2}
\end{figure}
\begin{remark}
When only two confidential messages are transmitted in the two-user frequency-selective BCC, the set of the s.d.o.f. has the equivalent MIMO interpretation \cite{lee2007hsa}. More specifically, the frequency-selective BCC (\ref{BCCwKusers}) is equivalent to the MIMO-BCC where the transmitter with $N+L$ dimensions (antennas) sends two confidential messages to two receivers with $N$ antennas. The secrecy constraint consumes $N$ dimensions for each MIMO link and lets the number of effective transmit antennas be $L$ for each user. The resulting channel is a two parallel $L\times N$ point-to-point MIMO channel without eavesdropper. Notice that the same parallel MIMO links can be obtained by applying the block diagonalization on the MIMO-BC without secrecy constraint \cite{lee2007hsa}. In other words, the secrecy constraint in the BCC with inner eavesdroppers is equivalent to the orthogonal constraint in the classical MIMO-BC.
Fig. \ref{fig:Interprete2} shows the example with $N=4,L=2$ and $K=2$ confidential messages.
\end{remark}

\section{Numerical Examples}\label{sec:numerical}
\begin{figure}[n]
    \centering
        \includegraphics[width=0.7\columnwidth]{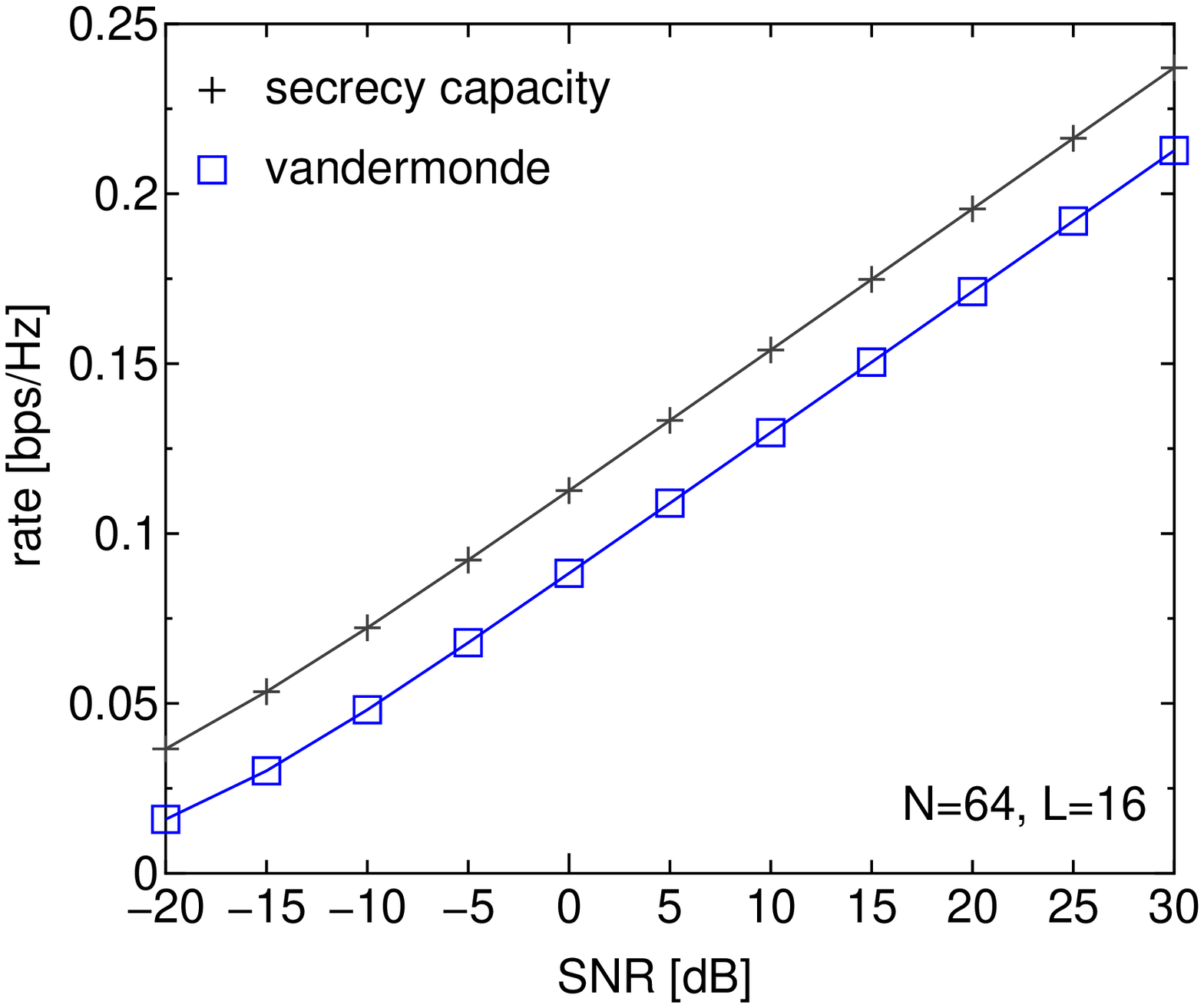}\\
        \vspace{-1.5em}
    \caption{Achievable secrecy rate with one observation at receiver 1 and $N=64, L=16$ (MISO wiretap channel).}
    \label{fig:CapacityMISO}
    \vspace{0.5em}
\end{figure}
\begin{figure}[n]
    \centering
    \includegraphics[width=0.7\columnwidth]{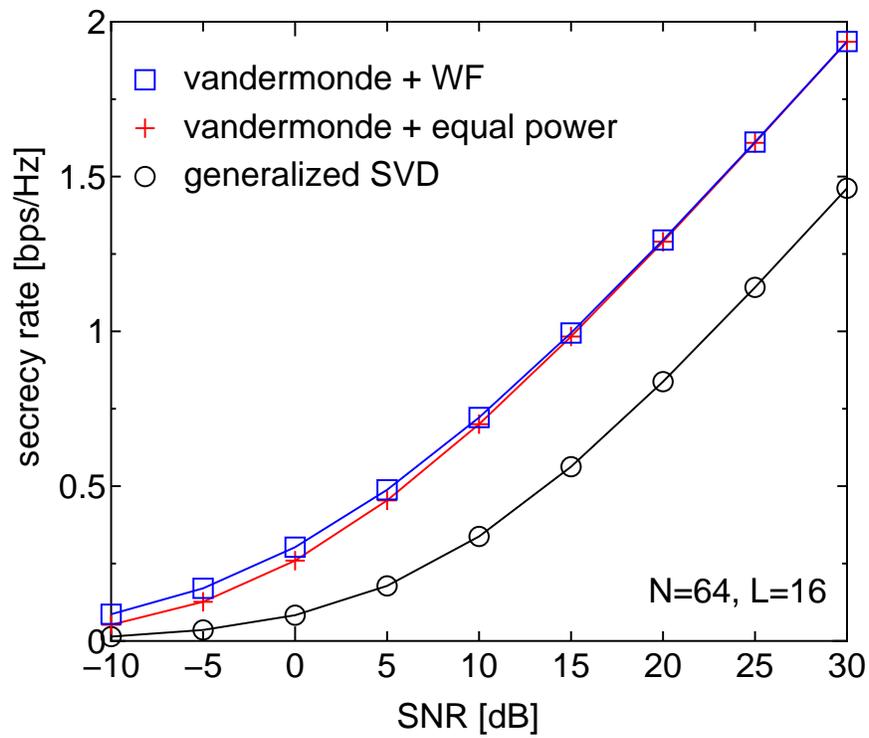}\\
    \vspace{-1.5em}
    \caption{Achievable secrecy rate with $N=64, L=16$ (MIMO wiretap channel).}
    \label{fig:N64L16}
    \vspace{0.5em}
\end{figure}
\begin{figure}[n]
    \centering
    \includegraphics[width=0.7\columnwidth]{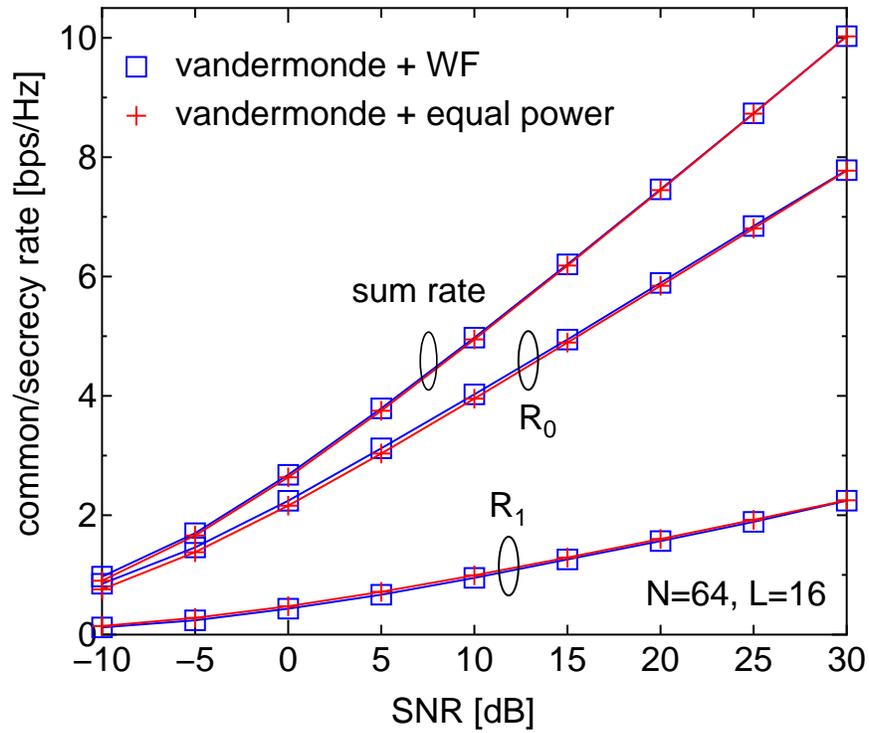}\\
    \vspace{-1.5em}
    \caption{Achievable secrecy/common rates $N=64, L=16$ in the frequency-selective BCC.}
    \label{fig:CommonSecrecyRate}
    \vspace{0.5em}
\end{figure}

\begin{figure}[n]
    \centering
    \includegraphics[width=0.7\columnwidth]{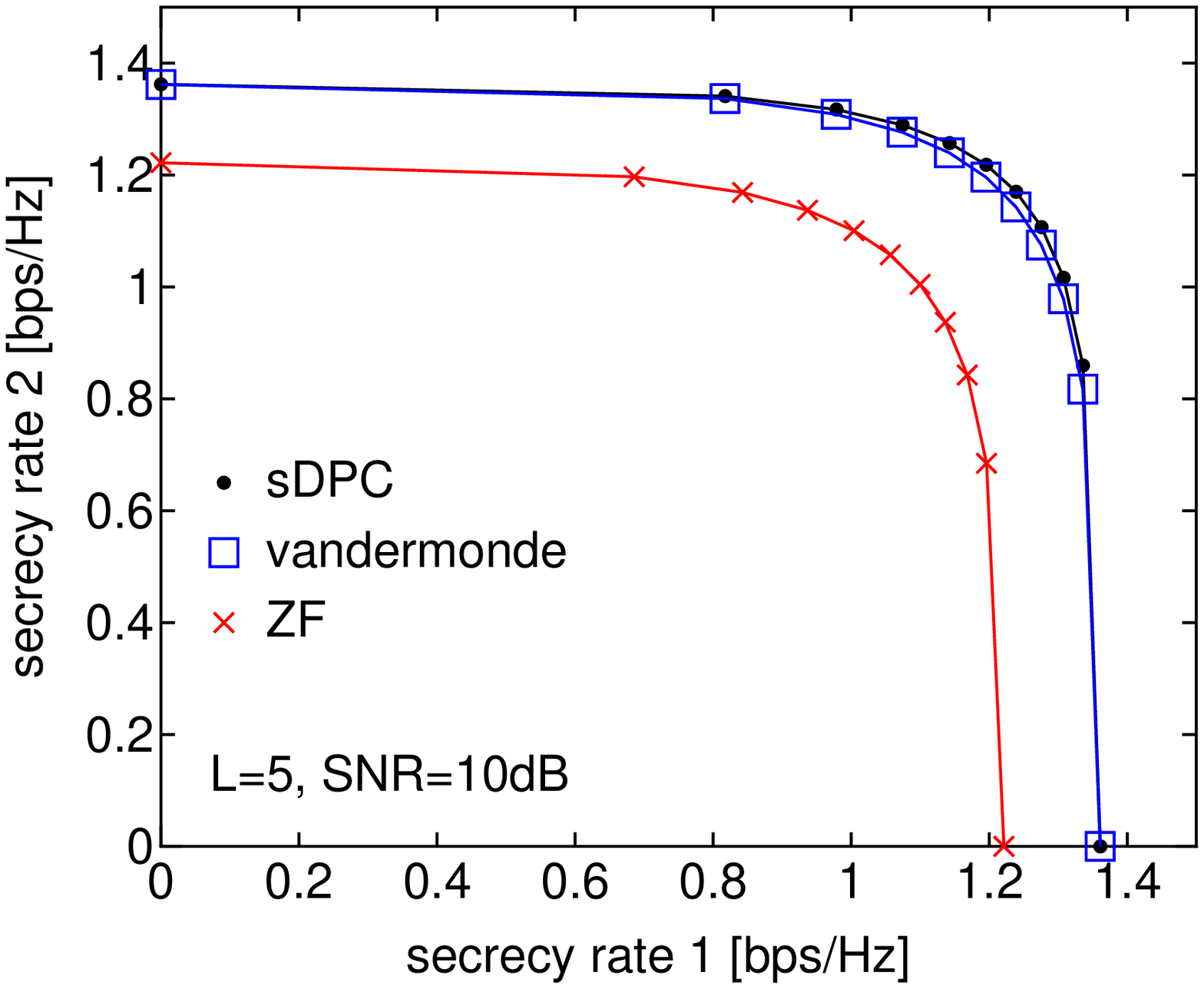}\\
    \vspace{-1.5em}
    \caption{Achievable secrecy rate region $N=1, L=5$ (MISO-BCC).}
    \label{fig:BCCregion}
    \vspace{0.5em}
\end{figure}
In order to examine the performance of the proposed Vandermonde precoding, this section provides some numerical results in different settings.

\subsection{Secrecy rate vs. SNR}
We evaluate the achievable secrecy rate $R_1^{\rm vdm}$ in (\ref{secrecyrate}) when the transmitter sends only a confidential message to receiver 1 (without a common message) in the presence of receiver 2 (eavesdropper) over the frequency-selective BCC studied in Section \ref{sec:vandermonde}.

\paragraph{MISO wiretap channel}
For the sake of comparison (albeit unrealistic), we consider the special case of the frequency-selective wiretap channel when receiver 1 has a scalar observation and the eavesdropper has $N$ observations. This is equivalent to the MISO wiretap channel with the receiver 1 channel $\hv\in\CC^{1\times (N+L)}$ and the eavesdropper channel $\Tc(\gv)\in\CC^{N\times (N+L)}$.
Without loss of generality, we assume that the observation at receiver 1 is the first row of $\Tc(\hv)$. We consider that all entries of $\hv, \gv$ are i.i.d. $\sim\Nc_{\Cc}(0, \frac{1}{L+1})$ and average the secrecy rate over a large number
of randomly generated channels with $N=64, L=16$. In Fig.
\ref{fig:CapacityMISO}, we compare the optimal beamforming strategy \cite{khisti2007stm,khisti:gmw,shafiee2007arg} and the Vandermonde precoding as a function of SNR $P$. Since only one stream is sent to receiver 1, the s.d.o.f. is $\frac{1}{N+L}$. In fact, the MISO secrecy
capacity in the high SNR regime is given by
\begin{equation}\label{Cmiso}
\frac{1}{N+L} \log\left(1+(N+L)P \max_{\phiv:
\Tc(\gv)\phiv=\zerov}|\hv \phiv |^2\right)
\end{equation}
where $\phiv\in\CC^{(N+L)\times 1}$ is the beamforming vector.
The Vandermonde precoding achieves
\begin{equation}\label{Cv-miso}
\frac{1}{N+L} \log\left(1+ (N+L)P \max_{i=1,\dots,L}|\hv \vv_{1,i} |^2\right).
\end{equation}
where $\vv_{1,i}$ denotes the $i$-th column of $\Vm_1\in\CC^{(N+L)\times L}$ orthogonal to $\Tc(\gv)$. Clearly, there exists a constant gap between (\ref{Cmiso}) and
(\ref{Cv-miso}) due to the suboptimal choice of the beamforming vector.

\paragraph{MIMO wiretap channel}
We consider the frequency-selective wiretap channel with $N=64, L=16$.
Although there exists a closed-form expression under a power-covariance constraint \cite{immse09}, the secrecy capacity under a total power constraint in (\ref{MIMOMEcapacity}) is still difficult to compute (especially for a large dimension of $N$ and $L$) because it requires a search over all possible power covariances constraints.
Therefore, in Fig. \ref{fig:N64L16}, we compare the averaged secrecy rate achieved by the generalized SVD scheme \cite{khisti2007mc} and the Vandermonde precoding.
We assume that all entries of $\hv, \gv$ are i.i.d. $\sim\Nc_{\Cc}(0, \frac{1}{L+1})$.
For the Vandermonde precoding, we show the achievable rate with waterfilling power allocation (\ref{secrecyrate}) and equal power allocation (\ref{EqualPowerSecRate}) by allocating $p=\frac{(N+L)P}{L}$ to $L$ streams. As observed, these two suboptimal schemes achieve the same s.d.o.f. of $\frac{L}{N+L}=\frac{1}{5}$ although the generalized SVD incurs a substantial power loss. The result agrees well with Theorem \ref{theorem:DoF-2user}. We remark also that the optimal waterfilling power allocation yields a negligible gain.

\subsection{The maximum sum rate point $(R_0,R_1)$ vs. SNR}
We consider the frequency-selective BCC with one confidential message to receiver 1 and one common message to two receivers. In particular, we characterize the maximum sum rate-tuple corresponding to $\gamma_0=\gamma_1$ on the boundary of the achievable rate region $\Rc_s$.
Fig. \ref{fig:CommonSecrecyRate} shows the averaged maximum sum rate-tuple $(R_0,R_1)$ of the Vandermonde precoding both with optimal input covariance computed by the greedy algorithm and with equal power allocation. We remark that there is essentially no loss with the equal power allocation.

\subsection{Two-user secrecy rate region in the frequency-selective BCC}
We consider the two-user frequency-selective BCC where the transmitter sends two confidential messages (no common message) of Section \ref{subsec:BCCinner}.
For the sake of comparison (albeit unrealistic), we consider the special case of one observation $N=1$ at each receiver.
Notice that the two-user frequency-selective BCC is equivalent to the two-user MISO BCC with $\hv_1,\hv_2\in\CC^{1\times (L+1)}$ whose secrecy capacity region is achieved by the S-DPC scheme \cite{LiuPoorIT09}. The proposed Vandermonde precoding achieves
the secrecy rate region given by all possible rate-tuples $(R_1,R_2)$
\begin{eqnarray*}
&& R_1 \leq \frac{1}{L+1} \log(1 + p_1\max_{i=1,\dots,L} |\hv_1\vv_{1,i} |^2 ) \\
&& R_2 \leq \frac{1}{L+1} \log(1 + p_2\max_{i=1,\dots,L} |\hv_2\vv_{2,i} |^2  )
\end{eqnarray*}
satisfying $p_1+p_2=(L+1)P$ where $\vv_{1,i},\vv_{2,i}$ denotes the $i$-th column of $\Vm_1\in\CC^{(N+L)\times L}$ orthogonal to $\hv_2$, $\Vm_2\in\CC^{(N+L)\times L}$ orthogonal to $\hv_1$, respectively. Fig. \ref{fig:BCCregion}
compares the averaged secrecy rate region of the Vandermonde precoding, zero-forcing beamforming, and the optimal S-DPC scheme for $L=5$ where
all entries of $\hv_1, \hv_2$ are i.i.d. $\sim\Nc_{\Cc}(0, \frac{1}{L+1})$. As observed, the Vandermonde precoding achieves the near-optimal rate region. As the number of paths $L$ increases, the gap with respect to the S-DPC becomes smaller since the Vandermonde precoding tends to choose the optimal beamformer matched to the channels.

\section{Conclusions}\label{sec:conclusion}

We considered the secured communication over the frequency-selective channel by focusing on the frequency-selective BCC. In the case of a block transmission of $N$ symbols followed by a guard interval of $L$ symbols discarded at both receivers, the frequency-selective channel can be modeled as a $N\times (N+L)$ Toeplitz matrix. For this special type of MIMO channels, we proposed a practical yet order-optimal Vandermonde precoding which enables to send $l\leq L$ streams of the confidential messages and $N-l$ streams of the common messages simultaneously over a block of $N+L$ dimensions. The key idea here consists of exploiting the frequency dimension to ``hide" confidential information in the zeros of the channel seen by the unintended receiver similarly to the spatial beamforming.
We also provided some application of the Vandermonde precoding in the multiuser secured communication scenarios and proved the optimality of the proposed scheme in terms of the achievable s.d.o.f. region.

We conclude this paper by noticing that there exists a simple approach to establish secured communications. More specifically, perfect secrecy can be built in two separated blocks ; 1) a precoding that cancels the channel seen by the eavesdropper to fulfill the equivocation requirement,  2) the powerful off-the-shelf encoding techniques to achieve the secrecy rate. Since the practical implementation of secrecy encoding techniques (double binning) remains a formidable challenge, such design is of great interest for the  future secrecy systems.

\section*{Acknowledgment}
The work is supported by the
European Commission in the framework of the FP7 Network of
Excellence in Wireless COMmunications NEWCOM++. The work of M\'erouane Debbah is supported by Alcatel-Lucent within the Alcatel-Lucent Chair on Flexible Radio at Supelec. The authors wish to thank Yingbin Liang for helpful discussions, and the anonymous reviewers for
constructive comments.

\appendices
\section{Proof of Lemma \ref{lemma:rank}}\label{appendix:rank-bcc1}
In this appendix, we consider the rank of $\Tc(\hv) \Vm_1$ where $\Vm_1$ satisfies the orthogonality $\Tc(\gv)\Vm_1=\zerov$. By letting $\vv_{1,i}$ denote the $i$-th column of $\Vm_1$ we have $\Vm_1=[\vv_{1,1},\dots,\vv_{1,L}]$ for the case of $l=L$.
We define the matrix $\Gm$ orthogonal to $\Vm_1$ by appending
$L-l$ rows $\vv_{1,l+1}^H,\dots,\vv_{1,L}^H$ to $\Tc(\gv)$
\[ \Gm = \left[
           \begin{array}{c}
             \Tc(\gv)  \\
              \vv_{1,l+1}^H\\
              \vdots\\
              \vv_{1,L}^H
           \end{array}
         \right]
\]
Notice that all $N+L-l$ rows are linearly independent.
By definition of $\Vm_1$, it is not difficult to see that $\Gm$ and $\Vm_1^H$ form a complete set of basis for a $N+L$-dimensional linear space. Indeed for $l=L$ the matrix $\Gm$ reduces to $\Tc(\gv)$, while $l<L$, a subset of a projection matrix onto the null space of $\Tc(\gv)$ are appended to $\Tc(\gv)$. Hence $\Tc(\hv)$ can be expressed as
\begin{equation}
    \Tc(\hv) = \Hm_{G} + \Hm_{V} = \Am \Gm + \Bm \Vm_1^H
\end{equation}
where $\Hm_G$ is the projection of $\Tc(\hv)$ onto the row vectors of $\Gm$ with a $N\times (N+L-l)$ coefficient matrix $\Am$, $\Hm_V$ is the projection of $\Tc(\hv)$ onto the row vectors of $\Vm_1$ with a $N\times l$ coefficient matrix $\Bm$.
\begin{eqnarray*}
\rank(\Tc(\hv) \Vm_1) &=& \rank((\Am \Gm + \Bm \Vm_1^H)\Vm_1)\\
&\overset{\mathrm{(a)}}=& \rank(\Bm)\\
&\overset{\mathrm{(b)}}=& \rank(\Bm \Vm_1^H \Vm_1\Bm^H) =\rank(\Hm_V) \\
&\overset{\mathrm{(c)}}=& l
\end{eqnarray*}
where (a) follows from the orthogonality $\Gm\Vm_1=\zerov$ and $\Vm_1^H\Vm_1=\Id_l$, (b) follows from $\rank(\Bm\Bm^H)=\rank(\Bm)$. The equality (c) is obtained as follows. We notice
\begin{eqnarray}\label{rankhg}
    \rank\left[
           \begin{array}{c}
             \Hm_V  \\
             \Gm
           \end{array}
         \right] \overset{\mathrm{(d)}}= \rank\left[
           \begin{array}{c}
             \Hm_V  +\Hm_G\\
             \Gm
           \end{array}
         \right]= \rank\left[
           \begin{array}{c}
             \Tc(\hv)\\
             \Gm
           \end{array}
         \right] \overset{\mathrm{(e)}}= \min(N+L,2N+L-l)= N+L
\end{eqnarray}
where in (d) adding $\Hm_G$ does not change the rank, (e) follows because any set of $N+L$ rows taken from $\Tc(\hv), \Gm$ are linearly independent (from the assumption that $\hv$, $\gv$ are linearly independent). Since $\Hm_V$ is orthogonal to $\Gm$, (\ref{rankhg}) yields
\[\rank(\Hm_V) = N+L-(N+L-l) = l\]
which establishes (c).
\section{Proof of Theorem \ref{theorem:covariance}} \label{appendix:covariance}
We consider the following three cases given in Lemma \ref{maxmin}.

\textbf{Case 1} Supposing $R_{01}<R_{02}$, we consider the objective function $f_1$ in (\ref{f1}). The objective
is concave only when $ \gamma_1\geq \gamma_0 $. Nevertheless, we consider the KKT conditions which are necessary for the optimality. It can be easily shown that
the KKT conditions are given by (\ref{case1-1}) and (\ref{case1-2})
where $\Psim_i \succeq \zerov$ is the Lagrangian dual matrix associated to the positive semidefiniteness constraint of $\Sm_i$ for $i=0,1$ and $\mu\geq 0$ is the Lagrangian dual variable associated to the total power constraint. It clearly appears that
for $\gamma_1\geq \gamma_0$ the objective is concave in $\Sm_0,\Sm_1$ and the problem at hand is convex. In this case, any convex optimization algorithm, the gradient-based algorithm \cite{viswanathan2003dce} for example, can be applied to find the optimal solution
while the algorithm converges to a local optimal solution for $\gamma_1< \gamma_0$.

\textbf{Case 2 :} Supposing $R_{02}<R_{01}$, we consider the objective function $f_2$ in (\ref{f2}).
Since the problem is convex ($f_2$ is concave and the constraint is linear in $\Sm_0,\Sm_1$), the KKT conditions are necessary and sufficient for optimality. We form the Lagrangian and obtain the following KKT conditions
\begin{eqnarray*}
&&\gamma_0  \Gm_0^H \left(\Id_N +  \Gm_0 \Sm_0\Gm_0^H\right)^{-1}\Gm_0 + \Psim_0  =  \mu\Id_{N+L-l}\\
&& \gamma_1 \Hm_1^H\left(\Id_N +\Hm_1\Sm_1\Hm_1^H\right)^{-1}\Hm_1+ \Psim_0 = \mu\Id_{l}\\
 && \trace(\Sm_0) + \trace(\Sm_1) = \overline{P} \\
&& \trace(\Psim_i \Sm_i) = 0, \;\;\; i=0,1
\end{eqnarray*}
where $\Psim_i \succeq \zerov$ is the Lagrangian dual matrix associated to the positive semidefiniteness constraint of $\Sm_i$ for $i=0,1$ and $\mu\geq 0$ is the Lagrangian dual variable associated to the total power constraint.
By creating $N$ parallel channels via SVD on $\Gm_0,\Hm_1$ in (\ref{SVD}), we readily obtain the solution (\ref{solution2}).

\textbf{Case 3} For $0< \theta<1$, we consider the objective function $f_3$ in (\ref{f3}).
In the following we focus on $\gamma_0>0$. Notice that if $\gamma_0=0$ we have $R_{01}=R_{02}=0$ which yields the corner point $(0,R_1^{\rm vdm})$ where $R_1^{\rm vdm}$ denotes the secrecy rate characterized in (\ref{secrecyrate}). The KKT conditions, necessary for the optimality,  are given by (\ref{case3-1}) and (\ref{case3-2})
where 
$\Psim_i \succeq \zerov$ is the Lagrangian dual matrix associated to the positive semidefiniteness  constraints for $i=0,1$ and $\mu\geq 0$ is the Lagrangian dual variable associated to the total power constraint.
The gradient-based algorithm \cite{viswanathan2003dce} can be applied to find the solution satisfying these KKT conditions.
Although this algorithm yields the optimal and unique solution
for $\gamma_1\geq \gamma_0\theta$, the algorithm converges to a local optimal solution for $\gamma_1<\gamma_0\theta$.
\section{Proof of Lemma \ref{lemma:region}}\label{appendix:achievability}
This section provides the achievability proof of the region (\ref{RegionKuserBCC}).
We extensively use the notation $A_{\epsilon}^{(n)} (P_{X,Y})$ to denote a set of jointly typical sequences $\xv,\yv$ of length $n$ with respect to the distribution $P(x,y)$. We let $\epsilon>0$ which is small for a large $n$.

\paragraph{Codebook generation}
Fix $P(u), P(v_1|u), \dots, P(v_K|u)$ and $P(x|u,v_1,\dots,v_K)$.
The stochastic encoder randomly generates
\begin{itemize}
  \item a typical sequence $\uv(w_0)$ according $P(\uv^n)=\prod_{i=1}^n P(u_i)$ where $w_0\in \{1,\dots,2^{n(R_0-\epsilon)}\}$
  \item for each $\uv^n$ and $k=1,\dots, K$, $2^{n(I(V_k;Y_k|U)-\epsilon)}$ i.i.d. codewords $\vv_k(w_k, j_k)$ each with $P(\vv_k^n|\uv^n)=\prod_{i=1}^n P(v_k^n|u_i)$, where the indices are given by
  \[w_k\in \{1,\dots,2^{n(R_1-\epsilon)}\}, j_k\in\{1, \dots, 2^{n(I(V_k;Z|U)-\epsilon) }\}. \]
\end{itemize}
Next, distribute the $2^{n(I(V_k;Y_k|U)-\epsilon)}$ codewords into $2^{nR_k}$ bins such that each bin
contains $2^{n (I(V_k;Z|U)-\epsilon)}$ codewords.

\paragraph{Encoding}
To send the messages $w_0,w_1,\dots,w_k$, we first choose randomly the index $w_0$ and the corresponding codeword $\uv(w_0)$.
Given the common message $\uv(w_0)$ and randomly chosen $K$ bins $w_1,\dots,w_K$, the encoder selects a set of indices
$j_1,\dots,j_K$ such that
\begin{equation}\label{TypicalSeq1}
\left(\uv(w_0), \vv_1(w_1, j_1), \dots, \vv_K(w_K, j_K) \right) \in A_{\epsilon}^{(n)} (P_{U, V_1,\dots,V_K})
\end{equation}
If there are more than one such sequence, it randomly selects one. Finally
the encoder selects $\xv$ according to $P(x|v_1,\dots,v_K)$.

\paragraph{Decoding} The received signals at the $K$ legitimate receivers are $\yv_1^n,\dots,\yv_K^n$, the outputs of
the channels $P(\yv_k^n |\xv^n) = \prod_{i=1}^n   P(y_k^n |x^n)$ for any $k$.
Receiver $k$ first chooses $w_{0}^{(k)}$ such that
\begin{equation}
(\uv(w_{0}^{(k)}),\yv_k)\in A_{\epsilon}^{(n)} (P_{U,Y_k})
\end{equation}
if such a $w_{0}^{(k)}$ exists. Then, for a given $w_{0}^{(k)}$ it chooses $w_k$ so that
\begin{equation}
(\uv(w_0^k), \vv_k(w_k,j_k),\yv_k)\in A_{\epsilon}^{(n)} (P_{U, V_k,Y_k})
\end{equation}
if such $\vv_k$ exists.

\paragraph{Error probability analysis} Without loss of generality, we assume that the message set is $w_0=w_1=\dots=w_k=1$. We remark that an error is declared if one or more of the following events occur.
\begin{itemize}
  \item  Encoding fails
      \begin{equation}
    E_1 \eqdef \{ \left(\uv(1), \vv_1(1, j_1), \dots, \vv_K(1, j_K) \right) \notin A_{\epsilon}^{(n)} (P_{U, V_1,\dots,V_K}) \}.
\end{equation}
From \cite{elgamal1981pms} we have $P(E_1)\leq \epsilon$, if
 \begin{equation}
 \sum_{k=1}^K R_k \leq  \sum_{k=1}^K I(V_k;Y_k|U) - \sum_{j=2}^K I(V_{\pi(j)}; V_{\pi(1)},\dots,V_{\pi(j-1)}|U)-I(V_1,\dots,V_K;Z|U)
\end{equation}
  \item Decoding step 1 fails ; there does not exist a jointly typical sequence for some $k$, i.e.
        \begin{equation}
    E_{2}^k \eqdef \{ \left(\uv(1), \vv_k(1, j_1), \yv_k \right) \notin A_{\epsilon}^{(n)} (P_{U, V_k,Y_k}) \}.
\end{equation}
From joint typicality \cite{cover_book} we have $P(E_{2}^k)\leq \epsilon$ for any $k$.
  \item Decoding step 2 fails ; there exits other sequences  satisfying the joint typicality for some $k$
      \begin{equation}
    E_{3}^k \eqdef \{ \forall(w_0,w_k)\neq (1,1),  \left(\uv(w_0), \vv_k(w_k, j_1), \yv_k \right) \in A_{\epsilon}^{(n)} (P_{U, V_k,Y_k}) \}.
\end{equation}
From \cite{elgamal1981pms} we have
$P(E_{3}^k)\leq \epsilon$ if $R_k\leq I(V_k;Y_k|U) - I(V_k;Z|U) -\epsilon$ for any $k$.
\end{itemize}
Hence, the error probability $P_e^{(n)}=P(E_1 \bigcup  (\cup_k E_{2k}) \bigcup (\cup_k E_{3k}) )\leq \epsilon $ if the rate-tuple satisfies (\ref{Region}).

\paragraph{Equivocation calculation}
To prove the equivocation requirement
\begin{eqnarray}
\sum_{k\in \Kc} R_k -\frac{1}{n} H(W_{\Kc}|Z^n) \leq \frac{\epsilon}{n}, \;\; \Kc\subseteq \{1,\dots,K\}
\end{eqnarray}
we consider $H(W_{\Kc}|Z^n) $ where we denote $W_{\Kc}=\{W_k,k\in \Kc\}$. Following the foot steps as the proof of \cite[Theorem 1]{bagherikaram2008srr}, we have
\begin{eqnarray*}
    H(W_{\Kc}|Z^n)&\overset{\mathrm{(a)}}\geq & H(W_{\Kc}|Z^n, U^n)\\
    &= & H(W_{\Kc},Z^n|U^n) - H(Z^n|U^n) \\
    &=& H(W_{\Kc},V_{\Kc}^n,Z^n|U^n)-H(V_{\Kc}^n|W_{\Kc},Z^n, U^n) - H(Z^n|U^n)\\
    &=& H(W_{\Kc},V_{\Kc}^n|U^n)+H(Z^n| W_{\Kc},V_{\Kc}^n,U^n)-H(V_{\Kc}^n|W_{\Kc},Z^n,U^n) - H(Z^n|U^n)\\
    &\overset{\mathrm{(b)}}\geq &H(W_{\Kc},V_{\Kc}^n|U^n)+H(Z^n| W_{\Kc},V_{\Kc}^n,U^n)-n\epsilon' - H(Z^n) \\
    &\overset{\mathrm{(c)}}= &H(W_{\Kc},V_{\Kc}^n|U^n)+H(Z^n| V_{\Kc}^n,U^n)-n\epsilon' - H(Z^n|U^n)\\
    &\overset{\mathrm{(d)}}\geq & H(V_{\Kc}^n|U^n) +H(Z^n| V_{\Kc}^n,U^n)-n\epsilon' - H(Z^n|U^n) \\
    &= & H(V_{\Kc}^n|U^n) -I(V_{\Kc}^n; Z^n|U^n)-n\epsilon' \\
   &\overset{\mathrm{(e)}}= & \sum_{k\in\Kc}H(V_k^n|U^n)-\sum_{j=2}^{|\Kc|} I(V_{\pi(j)}^n; V_{\pi(1)}^n,\dots,V_{\pi(j-1)}^n|U^n) -I(V_{\Kc}^n; Z^n|U^n)-n\epsilon'\\
    &\overset{\mathrm{(f)}}\geq & \sum_{k\in\Kc}I(V_k^n;Y_k^n|U^n)-\sum_{j=2}^{|\Kc|} I(V_{\pi(j)}^n; V_{\pi(1)}^n,\dots,V_{\pi(j-1)}^n|U^n) -I(V_{\Kc}^n; Z^n|U^n)-n\epsilon'\\
    &\geq & n \sum_{k\in\Kc} R_k -n\epsilon'
\end{eqnarray*}
where (a) follows because the conditioning decrease the entropy,
(b) follows from Fano's inequality \cite{cover_book} stating that for a sufficiently large $n$ we have
\[ H(V_{\Kc}^n|W_{\Kc},Z^n,U^n) \leq 1 + n P_{e,\rm eav}^{(n)} \sum_{k\in \Kc}I(V_k;Z|U) \leq n\epsilon'\]
where $P_{e,\rm eav}^{(n)}$ denotes the eavesdropper's error probability when decoding $V_{\Kc}^n$ with side information on $W_{\Kc}$ at the rate $\sum_{k\in \Kc}I(V_k;Z|U)$. It can be easily shown that $P_{e,\rm eav}^{(n)}\rightarrow 0$ as $n\rightarrow\infty$.
 (c) follows from the Markov chain $W_{\Kc}\rightarrow V_{\Kc}^n \rightarrow Z^n$, (d) follows
 by ignoring a non-negative term $ H(W_{\Kc}|V_{\Kc}^n)$, (e) follows because $H(V_{\Kc}^n|U^n)=\sum_{k\in\Kc}H(V_k^n|U^n)-\sum_{j=2}^{|\Kc|} I(V_{\pi(j)}^n; V_{\pi(1)}^n,\dots,V_{\pi(j-1)}^n|U^n)$ for any permutation $\pi$ over the subset $\Kc$, (f) follows because $H(V_k^n|U^n)\geq I(V_k^n;Y_k^n|U^n)$ for any $k$.

\section{Proof of Theorem \ref{theorem:DoFouterBCC}} \label{appendix:SDOF1}
The achievability follows by extending Theorem \ref{theorem:DoF-2user} to the case of $K$ confidential messages. First we remark that as a straightforward extension of Lemma \ref{lemma:rank} the following lemma holds.
\begin{lemma} \label{lemma:rank2}
For $\sum_{k=1}^K l_k \leq L$, there exists a matrix $[\Vm_1,\dots,\Vm_K]$ with $\sum_{k=1}^K l_k$ orthonormal columns with size $N+L$ satisfying
\begin{eqnarray} \label{KorthProperty}
&&\Tc(\gv) \Vm_k = \zerov_{N\times l_k}, \;\;\; k=1,\dots,K \\ \label{rankh_kV_j}
&&\rank\left(\Tc(\hv_k) (\sum_{j\in \Kc}\Vm_j \Vm_j^H)\Tc(\hv_k)^H \right)= \sum_{j\in \Kc} l_j,  \;\; \; \forall \Kc\subseteq\{1,\dots,K\}
\end{eqnarray}
where $l_k$ denotes the number of columns of $\Vm_k$
\end{lemma}
A sketch of proof is given in Appendix \ref{appendix:rank-bcc2}.

We let $\Vm_0$ be unitary matrix with $N+L-\sum_{k=1}^K l_k$ orthonormal columns in the null space of $[\Vm_1,\dots,\Vm_K]$ such that $\Vm_0^H  [\Vm_1,\dots,\Vm_K]= \zerov$. In other words, the Vandermonde precoder $\Vm=[\Vm_0,\dots,\Vm_K]$ is a squared unitary matrix satisfying $\Vm\Vm^H =\Id_{N+L}$.
Based on the Vandermonde precoder $\Vm$, we construct the transmit vector $\xv$ as
\begin{equation}\label{GSC}
\xv = \sum_{k=0}^K \Vm_k\uv_k
\end{equation}
where $\uv_0,\uv_1,\dots,\uv_K$ are mutually independent Gaussian vectors with zero mean and covariance $\Sm_0,\Sm_1,\dots,\Sm_K$ satisfying $\sum_{i=0}^K \trace(\Sm_i)\leq (N+L)P$.
From the orthogonality properties (\ref{KorthProperty}), the received signals become
\begin{eqnarray*}
\yv_k &=& \Tc(\hv_k) \Vm_0\uv_0 +  \Tc(\hv_k) \Vm_k\uv_k + \Tc(\hv_k)\sum_{j\neq k} \Vm_j\uv_j +\nv_k, \;\; k=1,\dots,K \\
\zv &=& \Tc(\gv)\Vm_0\uv_0 + \nuv
\end{eqnarray*}
where receiver $k$ observes the common message, the intended confidential message, and the interference from other users, while receiver $K+1$ observes only the common message. By letting $U=\Vm_0\uv_0$, $V_k=U + \Vm_k\uv_k$ for $k=1,\dots,K$, $X = U + \sum_{k=1}^K V_k$ and considering the equal power allocation to all $N+L$ streams, we readily obtain
\begin{eqnarray}
 I(U;Y_k)&=& \frac{1}{N+L}\log\frac{|\Id_N +P\Tc(\hv_k)\left(\sum_{j=0}^K{\Vm}_j {\Vm}_j^H\right)\Tc(\hv_k)^H |}
{|\Id_N +P\Tc(\hv_k)\left(\sum_{j=1}^K{\Vm}_j{\Vm}_j^H\right)\Tc(\hv_k)^H|},\;\; \forall k\\
I(U;Z) &=&  \frac{1}{N+L}\log\left|\Id_N +  P\Tc(\gv) {\Vm}_0 {\Vm}_0^H \Tc(\gv)^H\right|\\ \label{I(vk;yk)}
I(V_k;Y_k|U) &=&  \frac{1}{N+L}\log\frac{\left|\Id_N + P\Tc(\hv_k)(\sum_{j=1}^K{\Vm}_j{\Vm}_j^H) \Tc(\hv_k)^H\right|}{\left|\Id_N + P\Tc(\hv_k)(\sum_{j=1,j\neq k}^K{\Vm}_j {\Vm}_j^H ) \Tc(\hv_k)^H\right|},\;\; \forall k \\ \label{Penalty}
I(V_{\Kc};Z|U) &=&  0, \;\;\;\; \forall \Kc\subseteq \{1,\dots,K\}
\end{eqnarray}
and we also have $H(V_{\Kc}|U)=\sum_{k\in \Kc} H(V_k|U)$ from the independency between $V_1,\dots,V_K$ conditioned on $U$. Plugging this together with (\ref{I(vk;yk)}) and (\ref{Penalty}) into (\ref{RegionKuserBCC}), we have
\begin{eqnarray*}
   && R_k \leq I(V_k;Y_k|U), \;\; k=1,\dots,K \\
   && \sum_{k\in \Kc} R_k \leq \sum_{k\in \Kc}I(V_k;Y_k|U).
\end{eqnarray*}
In order to find the d.o.f. region, we notice
\begin{eqnarray}\nonumber
\rank(\Tc(\gv) \Vm_0 \Vm_0^H\Tc(\gv)^H) &=& \rank(\Tc(\gv)\Vm_0) \\ \nonumber
&\overset{\mathrm{(a)}}=& \rank(\Tc(\gv)[\Vm_0 \Vm_1,\dots,\Vm_K]) \\ \label{common}
&\overset{\mathrm{(b)}}=& \rank(\Tc(\gv))=N\\ \nonumber
\rank\left(\Tc(\hv_k)(\sum_{j=0}^K\Vm_j\Vm_j^H)\Tc(\hv_k)^H\right) &=& \rank(\Tc(\hv_k)\Vm\Vm^H\Tc(\hv_k)^H) \\ \label{enumerator1}
&\overset{\mathrm{(b)}}=& \rank(\Tc(\hv_k))=N \\  \label{denominator1}
\rank\left(\Tc(\hv_k)(\sum_{j=1}^K\Vm_j\Vm_j^H)\Tc(\hv_k)^H\right) &\overset{\mathrm{(c)}}=& \sum_{j=1}^K l_j\\ \label{denominator2}
\rank\left(\Tc(\hv_k)(\sum_{j=1,j\neq k}^K\Vm_j\Vm_j^H)\Tc(\hv_k)^H\right) &\overset{\mathrm{(c)}}=& \sum_{j=1,j\neq k}^K l_j
\end{eqnarray}
where (a) follows from orthogonality between $\Tc(\gv)$ and $\Vm_k$ for $k\geq 1$, (b)
follows from the fact that $\Vm=[\Vm_0\dots\Vm_K]$ is unitary satisfying $
\Vm\Vm^H=\Id$, and (c) follows from Lemma \ref{lemma:rank2}. From (\ref{enumerator1}) and (\ref{denominator1}), we readily
obtain
$r_0\leq \frac{N-\sum_{k=1}^K l_k}{N+L}$, which is dominated by (\ref{common}). Combining (\ref{denominator1}) and (\ref{denominator2}), we obtain $r_k\leq \frac{l_k}{N+L}$ for $k=1,\dots,K$. This completes the achievability.

The converse follows by a natural extension of Theorem \ref{theorem:DoF-2user} to the $K+1$-user BCC. To obtain the constraint (\ref{CooperativeBound1}), we consider that the first $K$ receivers perfectly cooperate to decode the $K$ confidential messages and one common message. By treating these $K$ receivers as a \emph{virtual} receiver with $KN$ antennas, we immediately obtain the bound (\ref{CooperativeBound1}) corresponding to the s.d.o.f. of the MIMO wiretap channel with the virtual receiver channel $[\Tc(\hv_1)^T,\dots,\Tc(\hv_K)^T]^T$ and the eavesdropper channel $\Tc(\gv)$. The bound (\ref{CooperativeBound2}) is obtained by noticing that the total number of streams that receiver $k$ can decode is limited by the d.o.f. of $\Tc(\hv_k)$, i.e. $N$. Namely, we have the following $K$ inequalities
\begin{equation}\label{MIMOconstK}
    l_0 + l_k \leq N, \;\; k=1,\dots,K
\end{equation}
which yields $l_0\leq N-\max_k l_k$. Further by letting $l_k=L$ for any $k\in\{1,\dots,K\}$ and
and $l_j=0$ for any $j\neq k$, we obtain $l_0\leq N-L$. Adding the last inequality and (\ref{CooperativeBound1}), we obtain (\ref{CooperativeBound2}). This establishes the converse.
\section{Proof of Lemma \ref{lemma:rank2}}\label{appendix:rank-bcc2}
We consider $\rank(\Tc(\hv_k)\sum_{j\in \Kc} \Vm_j\Vm_j^H \Tc(\hv_k)^H)$ for a subset $\Kc\subseteq\{1,\dots,K\}$.
First we let $\vv_{c,1},\dots,\vv_{c,L}$ denote $L$ orthonormal columns that form a unitary Vandermonde matrix orthogonal to $\Tc(\gv)$. For any subset $\Lc \subseteq\{1,\dots,L\}$, we let $\Vm_{c,\Lc}$ be the unitary matrix formed by $|\Lc|$ columns corresponding to the subset $\Lc$ taken from $\vv_{c,1},\dots,\vv_{c,L}$.
Since a unitary matrix formed by $\{\Vm_k\}_{k\in\Kc}$ for any $\Kc$ can be expressed equivalently as $\Vm_{c,\Lc}$,  we consider $\rank(\Tc(\hv_k)\Vm_{c,\Lc} \Vm_{c,\Lc}^H\Tc(\hv_k)^H)$. For a given $\Lc$, we let $\Vm_{c,\overline{\Lc}}$ denote a unitary matrix composed by $L-|\Lc|$ columns corresponding to the complementary set $\overline{\Lc}$ such that
$\Lc+\overline{\Lc}=\{1,\dots,L\}$. In order to derive the rank, we follow the same approach as Appendix \ref{appendix:rank-bcc1}. We define the matrix $\Gm_{\Lc}\in\CC^{(N+L-|\Lc|)\times (N+L)}$ orthogonal to $\Vm_{c,\Lc}$ by appending
$\Vm_{c,\overline{\Lc}}^H$ to $\Tc(\gv)$
\[ \Gm_{\Lc} = \left[
           \begin{array}{c}
             \Tc(\gv)  \\
              \Vm_{c,\overline{\Lc}}^H
           \end{array}
         \right]
\]
where the $N+L-|\Lc|$ rows are linearly independent. Since $\Gm_{\Lc}$ and $\Vm_{c,\Lc}^H$ form a complete set of a $N+L$-dimensional linear space, $\Tc(\hv_k)$ can be expressed as
\begin{equation}
    \Tc(\hv_k) = \Am_{k,\Lc}\Gm_{\Lc} + \Bm_{k,\Lc} \Vm_{c,\Lc}^H, \;\; k=1,\dots,K
\end{equation}
where $\Am_{k,\Lc},\Bm_{k,\Lc}$ is a coefficient matrix with dimension $N\times (N+L-|\Lc|)$, $N\times |\Lc|$ respectively. By recalling that any set of $N+L$ rows taken from $\Tc(\hv_k),\Tc(\gv)$ are linearly independent for $k=1,\dots,K$ (from the assumption that $\gv,\hv_1,\dots,\hv_K$ are linearly independent), we can repeat the same argument as Appendix \ref{appendix:rank-bcc1} and obtain
\[\rank(\Tc(\hv_k)\Vm_{c,\Lc} \Vm_{c,\Lc}^H\Tc(\hv_k)^H) = |\Lc|, \;\; \forall \Lc \subseteq\{1,\dots,L\}, k=1,\dots,K \]
which yields the result.

\section{Proof of Theorem \ref{theorem:DoFinnerBC}} \label{appendix:SDOF2}
The achievability follows by generalizing Theorem \ref{theorem:DoF-2user} for the case of two confidential messages. We remark that by symmetry Lemma \ref{lemma:rank} for one beamforming matrix $\Vm_1$ can be trivially extended to two beamforming matrices $\Vm_1$ and $\Vm_2$. Namely, we have
\begin{lemma}
For $l_1\leq L$ and $l_2\leq L$, there exists $\Vm_k$ with $l_k$ orthnormal columns for $k=1,2$ satisfying
\begin{eqnarray}\label{OrthoNOrtho}
&& \Tc(\hv_k) \Vm_j =\zerov_{N\times l_j}, \;\; k=1,2, j\neq k \\ \label{rankhkvj}
&& \rank(\Tc(\hv_k) \Vm_k) = l_k, \;\; k=1,2
\end{eqnarray}
\end{lemma}

Further, we let $\Vm_0$ be a unitary matrix with $M=N+L-\rank([\Vm_1 \Vm_2])$ orthonormal columns in the null space of $[\Vm_1\Vm_2]$ such that $\Vm_0^H [\Vm_1\Vm_2]=\zerov_{M\times (l_1+l_2)}$.
We construct $\xv$ by Gaussian superposition coding based on the the Vandermonde precoder $\Vm_0, \Vm_1$ and $\Vm_2$. From (\ref{OrthoNOrtho}), each user observes the vector of its confidential message and that of the common message, i.e.
\begin{eqnarray} \label{ReceivedSig}
\yv_1 &= & \Tc(\hv_1) (\Vm_0 \uv_0+ \Vm_1 \uv_1) + \nv_1 \\ \nonumber
\yv_2 &= & \Tc(\hv_2) (\Vm_0 \uv_0+ \Vm_2 \uv_2) + \nv_2
\end{eqnarray}
By letting $U = \Vm_0\uv_0, V_k=U+ \Vm_k\uv_k$ for $k=1,2$, $X=V_1+V_2$ and considering equal power allocation to all streams with $p=\frac{(N+L)P}{M+l_1+l_2}$, we readily obtain
\begin{eqnarray*}
I(U;Y_k)&=&\frac{1}{N+L} \log\frac{|\Id_N + p\Tc(\hv_k)(\Vm_0\Vm_0^H+\Vm_k\Vm_k^H)\Tc(\hv_k)^H| }{|\Id_N + p \Tc(\hv_k)\Vm_k\Vm_k^H\Tc(\hv_k)^H|} \\ \nonumber
 I(V_k;Y_k|U) &=& \frac{1}{N+L} \log|\Id_N + p\Tc(\hv_k)\Vm_k\Vm_k^H\Tc(\hv_k)^H|  \\
 I(V_1;Y_2,V_2|U) &=&I(V_2;Y_1,V_1|U)=0
\end{eqnarray*}
We remark
\begin{eqnarray*}
\rank(\Tc(\hv_k)\Vm_k\Vm_k^H\Tc(\hv_k)^H)&=&\rank(\Tc(\hv_k)\Vm_k) = l_k, \;\; k=1,2 \\
 \rank(\Tc(\hv_k)(\Vm_0\Vm_0^H+\Vm_k\Vm_k^H)\Tc(\hv_k)^H) &=&  \rank\left(\Tc(\hv_k)[\Vm_0 \Vm_k] \left[\begin{array}{c}
           \Vm_0^H \\
           \Vm_k^H
         \end{array}\right]\Tc(\hv_k)^H\right) \\
         &\overset{\mathrm{(a)}}=& \rank(\Tc(\hv_k)[\Vm_0 \Vm_k \Vm_j]) \\
          &\overset{\mathrm{(b)}}=& \rank(\Tc(\hv_k))=N
\end{eqnarray*}
where (a) follows from orthogonality between $\Tc(\hv_k)$ and $\Vm_j$ for $j\neq k$, (b)
follows because $[\Vm_0\Vm_1\Vm_2]$ or $[\Vm_0\Vm_2\Vm_1]$ spans a complete $N+L$-dimensional space. These equations yield
$l_0+l_k\leq N$ for $k=1,2$. This establishes the achievability.

The converse follows by noticing that the constraints (\ref{peruserSDoF}) and (\ref{MIMOconst}) correspond to trivial upper bounds. To obtain (\ref{peruserSDoF}), we consider the special case when the transmitter sends only one confidential message to one of two receivers in the presence of the eavesdropper. When sending one confidential message to receiver 1, the two-user frequency-selective BCC reduces to the MIMO wiretap channel with the legitimate channel $\Tc(\hv_1)$ and the eavesdropper channel $\Tc(\hv_2)$, whose s.d.o.f. is upper bounded by $L$. The same bound holds for receiver 2 when transmitting one confidential message to receiver 2 in the presence the eavesdropper (receiver 1).
The upper bounds (\ref{MIMOconst}) follow because the total number of streams per receiver is limited by the individual $(N+L)\times N$ MIMO link. This establishes the converse.

\bibliographystyle{IEEEbib}
\bibliography{VandermondeWiretapJournalver7}
\end{document}